\documentclass[11pt,letterpaper]{article}

\usepackage[margin=1in]{geometry}\usepackage{graphicx} 
\usepackage{subcaption}
\usepackage{hyperref} 
\usepackage{amsmath}
\usepackage{amssymb}
\usepackage{amsfonts}
\usepackage{amsthm,thmtools}
\usepackage{fullpage}
\usepackage{enumitem}
\usepackage{algorithm}
\usepackage{xcolor}
\usepackage{booktabs}
\usepackage{cite}
\usepackage[capitalise,noabbrev]{cleveref}

\crefname{line}{Line}{Lines}
\Crefname{line}{Line}{Lines}
\usepackage[colorinlistoftodos]{todonotes}
\setuptodonotes{backgroundcolor=blue!10, bordercolor=blue!10, linecolor=blue!10, textcolor=blue!50}

\numberwithin{equation}{section}

\usepackage{algpseudocode}
\algrenewcommand\algorithmicrequire{\textbf{Input:}}
\algrenewcommand\algorithmicensure{\textbf{Output:}}

\declaretheorem[style=plain,numberwithin=section,name=Theorem]{theorem}
\declaretheorem[style=plain,sibling=theorem,name=Lemma]{lemma}
\declaretheorem[style=plain,sibling=theorem,name=Proposition]{proposition}
\declaretheorem[style=plain,sibling=theorem,name=Corollary]{corollary}

\declaretheorem[style=plain,sibling=theorem,name=Observation]{observation}

\declaretheorem[style=definition,sibling=theorem,name=Definition,qed=$\blacksquare$]{definition}
\declaretheorem[style=definition,sibling=theorem,name=Example,qed=$\blacksquare$]{example}

\declaretheorem[style=definition,sibling=theorem,name=Remark,qed=$\blacksquare$]{remark}

\crefname{theorem}{Theorem}{Theorems}
\crefname{proposition}{Proposition}{Propositions}
\crefname{lemma}{Lemma}{Lemmas}
\crefname{exmp}{Example}{Examples}
\crefname{corollary}{Corollary}{Corollaries}
\crefname{claim}{Claim}{Claims}
\crefname{remark}{Remark}{Remarks}
\crefname{section}{Section}{Sections}
\crefname{observation}{Observation}{Observations}
\crefname{definition}{Definition}{Definitions}

\DeclareMathOperator{\cl}{cl}
\newcommand{\M}{{\mathcal{M}}}
\newcommand{\I}{{\mathcal{I}}}
\newcommand{\B}{{\mathcal{B}}}
\newcommand{\Z}{{\mathbb{Z}}}

\DeclareMathOperator{\profit}{profit}
\DeclareMathOperator{\rank}{rank}

\title{A Matroid Generalization of Weighted Popular Matching: Characterization and Algorithm}

\author{
  Yuni Iwamasa\thanks{Kobe University, Japan, \href{mailto:iwamasa@math.kobe-u.ac.jp}{\texttt{iwamasa@math.kobe-u.ac.jp}}.}
  \and
  Yusuke Kobayashi\thanks{Kyoto University, Japan, \href{mailto:yusuke@kurims.kyoto-u.ac.jp}{\texttt{yusuke@kurims.kyoto-u.ac.jp}}.}
  \and
  Ryoga Mahara\thanks{The University of Tokyo, Japan, \href{mailto:mahara@mist.i.u-tokyo.ac.jp}{\texttt{mahara@mist.i.u-tokyo.ac.jp}}.}
  \and
  Kenjiro Takazawa\thanks{Hosei University, Japan, \href{mailto:takazawa@hosei.ac.jp}{\texttt{takazawa@hosei.ac.jp}}.}
}
\date{}

\begin{document}

\maketitle

\begin{abstract}
We study a problem that unifies two major generalizations of popular matchings in the one-sided preference model.
In a bipartite graph in which each vertex, or applicant, on one side has preferences over the adjacent vertices, 
a matching is called popular if it does not lose a head-to-head election against any other matching. 
Since the seminal work of Abraham, Irving, Kavitha, and Mehlhorn, 
who gave the first polynomial-time algorithm for finding a popular matching, 
numerous generalizations of the popular matching problem have been investigated.
Two prominent directions are weighted popular matchings, 
in which each applicant is assigned a voting weight, 
and popularity problems under matroid constraints. 
The former admits a polynomial-time algorithm due to Mestre, 
while the most general form of the latter is the popular common independent set problem,
for which a polynomial-time algorithm is given by Kavitha, Makino, Schlotter, and Yokoi.

In this paper, we introduce the weighted popular common independent set problem, 
which simultaneously generalizes weighted popular matchings and popular common independent sets. 
Our main result is the first polynomial-time algorithm for this problem. 
A key ingredient of the algorithm is a solution-independent characterization of weighted popular common independent sets,
which we establish in this paper.
As consequences of our algorithmic result, 
we obtain polynomial-time algorithms for several previously unresolved problems, 
including weighted popular matchings with matroid constraints and weighted popular branchings.
\end{abstract}

\thispagestyle{empty}

\newpage
\setcounter{page}{1}

\section{Introduction}

The study of \emph{matching problems with preferences} \cite{Man13} lies at the intersection of combinatorial optimization, algorithmic game theory, and market design. 
A classical solution concept in two-sided matching markets is \emph{stability}, which rules out local pairwise deviations in the form of blocking pairs. 
\emph{Popularity}, 
introduced by G\"{a}rdenfors~\cite{Gar75}, offers a different solution concept based on voting-theoretic optimality: a matching is popular if it does not lose a head-to-head election against any other matching.  
Therefore, 
a popular matching can be viewed as a matching that satisfies a Condorcet-type 
criterion and achieves a form of global optimality. 

A principal model of popular matching is the \emph{one-sided preference model}. 
In this model, 
we are given a bipartite graph $G=(V,U;E)$.
Each element $v\in V$ is called an {\em applicant} or a \emph{voter}, 
and each element $u \in U$ is called a \emph{post}. 
Each applicant $v\in V$ has a total order $\succ_v$ over the posts adjacent to $v$, 
or equivalently, over the edges incident to $v$, 
while the posts themselves have no preferences. 
Let $M,M'\subseteq E$ be matchings in $G$. 
We say that {\em $v \in V$ prefers $M$ to $M'$} if 
$M(v)\succ_v M'(v)$, 
where $M(v)$ (resp.\ $M'(v)$) denotes the edge incident to $v$ in $M$ (resp.\ $M'$), 
which may be empty. 
Let $\Delta(M, M')$ denote the difference in preferences between $M$ and $M'$, that is, 
\[
\Delta(M, M')
:= \left|\{v \in V\mid v \text{ prefers } M \text{ to } M' \}\right|
 - \left|\{v \in V\mid v \text{ prefers } M' \text{ to } M \}\right|.
\]
A matching $M$ is called \emph{popular} if 
$\Delta(M, M') \ge 0$ for all matchings $M'$ in $G$.




This model captures assignment problems such as the allocation of
houses, jobs, or public resources to applicants. 
In 2005, Abraham, Irving, Kavitha, and Mehlhorn initiated the algorithmic
study of popular matchings in this setting \cite{AIKM05,AIKM07}. 
They revealed a useful characterization of popular matchings in 
terms of the \emph{first} and \emph{second edges}, 
which led to the design of polynomial-time algorithms 
for finding a popular matching. 

These results stimulated 
the study of popular matchings for several generalized models and variants, 
including 
popular matchings under preferences with ties \cite{AIKM05,AIKM07}, 
capacitated popular matchings \cite{MS06,Pal14}, 
optimal popular matchings
\cite{BK19,FK22,FKPZ19,Kav14,Kav20,KN09,MI11}, 
and popular branchings \cite{KKMSS22b,NT23}.
The \emph{two-sided preference model} \cite{HK13} and the \emph{roommate model} \cite{BKK25,BIM10,Chu00} have also been studied, along with similar generalizations. 

In this paper, we focus on two extensions in the one-sided preferences setting: 
\emph{weighted popular matchings} \cite{Mes06,Mes14} and 
\emph{popular common independent sets}
\cite{CKTY25,Kam16,Kam17,Kam20,KMSY25}. 
Our main contribution is to provide the first polynomial-time algorithm for 
finding a \emph{weighted popular common independent set}, a model that simultaneously generalizes both of the above extensions.

\subsection{Our Result}


The weighted popular common independent set problem has two main features: 
(1) each voter has a voting power, represented by a weight; 
and 
(2) bipartite matching is generalized to matroid intersection. 
In this problem, 
we are given two matroids on the same ground set $E$; see Section~\ref{sec:matroid} for the definition of a matroid.
One of them is a partition matroid $\M_{\mathrm{part}} = (E, \mathcal{I}_{\mathrm{part}})$ defined by  
partition classes $E_v$ indexed by the elements $v$ of a finite set $V$. 
That is, $E = \bigsqcup_{v \in V} E_v$ and 
\[
\mathcal{I}_{\mathrm{part}} := \{ X\subseteq E \mid |X\cap E_v| \le 1 \text{ for all } v \in V \}.
\]
The other matroid is an arbitrary matroid $\M = (E, \mathcal{I})$, which is given by an independence oracle.
Each element $v\in V$ is called an {\em applicant} or a \emph{voter}. 
Each applicant $v\in V$ is assigned a positive weight $w(v) \in \mathbb{R}_{> 0}$, and has a total order $\succ_v$ over the elements of $E_v$. 
Note that this is in contrast to the model of optimal popular matchings, 
in which weights are assigned to the elements of $E$. 

Intuitively, 
$w(v)$ represents the voting power of $v$, 
and $\succ_v$ represents the preferences of $v$ over $E_v$. 
A popular common independent set is a common independent set of $\M_{\mathrm{part}}$ and $\M$ that does not lose to any other common independent set in a head-to-head election, 
which is formally defined as follows. 
For any pair of common independent sets $I$ and $I'$ of the matroids $\M_{\mathrm{part}}$ and $\M$, we say that {\em $v \in V$ prefers $I$ to $I'$} if either 
\begin{itemize}
    \item $I \cap E_v \neq \emptyset$ and $I' \cap E_v = \emptyset$, or
    \item $I \cap E_v = \{e\}$ and $I' \cap E_v = \{e'\}$ with $e \succ_v e'$.
\end{itemize}
For any pair of common independent sets $I$ and $I'$, let $\Delta_w(I, I')$ denote the weighted difference in preferences between $I$ and $I'$, that is, 
\[
\Delta_w(I, I')
:= \sum_{v \in V\colon v \text{ prefers } I \text{ to } I'} w(v)
 - \sum_{v \in V\colon v \text{ prefers } I' \text{ to } I} w(v) .
\]

\begin{definition}\label{def:popular}
    A common independent set $I$ of $\M_{\mathrm{part}}$ and $\M$ is \emph{popular} if $\Delta_w(I, I') \ge 0$ for all common independent sets $I'$.
\end{definition}

The \emph{weighted popular common independent set problem} is to find a popular common independent set, 
or determine that no such set exists.
The main contribution of this paper is to present the first polynomial-time algorithm for this problem. 


\begin{theorem}\label{thm:main}
 There exists a polynomial-time algorithm for the weighted popular common independent set problem. 
\end{theorem}

Table~1 summarizes known polynomial-time algorithms for popular matchings and their extensions, including the results obtained in this paper.
 
In particular, our algorithm yields the first polynomial-time algorithms for 
\emph{weighted popular matchings with matroid constraints} and for \emph{weighted popular branchings}, 
the design of which has been left open in \cite{Kam18} and \cite{NT23}, respectively; 
we discuss the details in the following subsections.

\begin{table}
\centering
\caption{Summary of polynomial-time algorithms for popular matchings and their extensions}
\label{tab:results}
\begin{tabular}{lcc}
\toprule
          & Unweighted & Weighted \\
\midrule
popular matching  & \cite{AIKM05,AIKM07} & \cite{Mes06,Mes14} \\
popular capacitated matching & \cite{MS06} & \cite{SM10} \\
popular matching with matroid constraints & \cite{Kam16} & {\bf this paper} \\
 & &  (partial result in~\cite{Kam18}) \\
popular branching & \cite{KKMSS22b} & {\bf this paper}\\
 & &  (partial result in~\cite{NT23}) \\
popular common independent set & \cite{KMSY25} & {\bf this paper} \\
\bottomrule
\end{tabular}
\end{table}

\subsection{Previous Work}



The weighted popular common independent set problem includes the following previous models as special cases. 

\paragraph{Weighted Popular Matchings.}

Weighted popular matchings are obtained by assigning weights to applicants, and constitute one of the most fundamental generalizations of popular matchings. 
Mestre \cite{Mes06,Mes14} introduced this model in 2006.
In the weighted setting, a matching is popular if no other matching receives a larger total weight of votes against it. Mestre \cite{Mes06,Mes14} gave 
a characterization of weighted popular matchings and a polynomial-time algorithm for finding such a matching. 
This work was later extended by Sng and Manlove \cite{SM10} to the capacitated matching setting.

\paragraph{Popular Matchings and Matroids.}

Several generalizations of popular matchings to matroid intersection have been investigated. 
Kamiyama \cite{Kam16} first introduced a model of \emph{popular matchings with matroid constraints}, 
which generalizes capacitated popular matchings. 
In this model, 
each post $u$ is endowed with 
a matroid $\M_u$ on the edges incident to $u$. 
An edge set $M$ is a solution if 
each applicant has at most one adjacent edge in $M$, 
and 
the set of edges in $M$ incident to each post $u$ is independent in $\M_u$. 
For this model, 
Kamiyama \cite{Kam16} extended the characterization and polynomial-time algorithm for popular matchings by \cite{AIKM05,AIKM07}.

Later, the model of \emph{popular branchings} is introduced by Kavitha, Kir\'{a}ly, Matuschke, Schlotter, and Schmidt-Kraepelin~\cite{KKMSS22b}, 
in which popularity is defined for the branchings in a digraph. 
They provided a characterization 
of popular branchings in terms of \emph{dual certificates}, 
and a polynomial-time algorithm for finding a popular branching. 

A direct connection between popular matchings and popular branchings had not been described explicitly 
until Kavitha, Makino, Schlotter, and Yokoi~\cite{KMSY25} presented a model of \emph{popular common independent sets}. 
This model discusses popularity of common independent sets in matroid intersection, 
and amounts to a common generalization of the models of 
popular matchings with matroid constraints and popular branchings. 
The authors presented a characterization and a polynomial-time algorithm for popular common independent sets \cite{KMSY25}. 

\paragraph{Weighted Popular Matchings and Matroids.}

Weighted extensions of the above matroidal generalizations have also been studied. 
Nevertheless, a comprehensive understanding of such weighted settings is still missing, 
especially from an algorithmic viewpoint.

\emph{Weighted popular matchings with matroid constraints} were studied by Kamiyama \cite{Kam18}, 
who provided a characterization, while the design of an efficient algorithm 
for finding a weighted popular matching in this model was left open. 
\emph{Weighted popular branchings} have been studied in \cite{NT23}, 
in which a polynomial-time algorithm was obtained only for a restricted case. 

\bigskip

It is not difficult to see that each of the above models is 
a special case of that of weighted popular common independent sets. 
Thus, 
our algorithm for weighted popular common independent sets 
is applicable to all of the above models, 
and in particular 
resolves the algorithmic questions for 
weighted popular matchings with matroid constraints and for weighted popular branchings.

\subsection{Technical Challenges}

We here describe our analysis of why these algorithmic questions had remained open and where the main difficulty lies.

\paragraph{Establishing a solution-independent characterization.}
The characterizations for the models discussed above can be broadly classified into two categories. 
In one category, 
the conditions are formulated without explicitly involving the candidate solution $M$. 
Such characterizations can be exploited algorithmically; 
for instance, the characterization in \cite{AIKM05,AIKM07} directly yields an efficient algorithm.

In the other category, 
the popularity of a solution $M$ is characterized by 
a necessary and sufficient condition that explicitly involves the solution $M$ itself. 
Such characterizations typically state that a solution $M$ is popular 
if and only if it is a maximum-weight solution with respect to edge weights 
defined from $M$. 
While this type of characterization can be used to determine whether a given solution $M$ is popular, 
it does not immediately yield an efficient algorithm. 
Therefore, 
establishing a solution-independent characterization is important for designing algorithms. 
A connection between these two types of characterizations is analyzed in \cite{KT25}. 

For weighted popular matchings, 
the characterization in \cite{Mes06,Mes14} is solution-independent, 
while that of weighted popular matchings with matroid constraints \cite{Kam18} is solution-dependent. 
This appears to be one of the main reasons why the algorithm design for weighted popular matchings with matroid constraints has remained open.

\paragraph{Well-formedness, pruning algorithm, and promotion paths.}



As mentioned above, no extension of the solution-independent characterization for weighted popular matchings \cite{Mes06,Mes14} 
to the setting with matroid constraints has been obtained so far.
One possible reason is that the solution-independent characterization in \cite{Mes06,Mes14} is substantially more involved than its counterpart for the unweighted setting \cite{AIKM05,AIKM07}.

More specifically, 
this involvement can be explained as follows.
Mestre \cite{Mes06,Mes14} introduced the notion of \emph{well-formed} solutions. 
This notion corresponds to the characterization of popularity for many of the unweighted settings, 
while 
it is only 
a necessary condition 
for the weighted settings \cite{Kam18,Mes06,Mes14,SM10}. 
Mestre's characterization is described by this well-formedness and a \emph{pruning algorithm} applied to the input graph.
After the pruning process, a matching is popular if and only if it is well-formed.

To establish the correctness of this pruning algorithm, 
Mestre \cite{Mes06,Mes14} introduced the notion of \emph{promotion paths}, 
which are used to characterize improvements of a matching in terms of popularity.
Although the concept of promotion paths has been extended in subsequent work~\cite{Kam18,SM10}, 
the pruning algorithm itself has proved difficult to generalize to settings with matroid constraints.

The technical difficulty of the pruning algorithm may be one reason why, 
despite weighted popular matchings constituting a natural and fundamental generalization of popular matchings, 
subsequent research on the subject has remained relatively limited.


\bigskip

The above analysis highlights the importance and difficulty of a characterization that does not depend on a given solution. 
One of our main contributions, alongside the algorithmic results, is precisely such a characterization for weighted popular common independent sets. This characterization is derived from several new technical ingredients, which we now outline.

\paragraph{Generalization of the well-formedness.}
We first extend the notion of well-formed solutions. 
For the weighted popular common independent set problem, we prove that our well-formedness indeed offers a necessary condition for popularity, 
extending 
the corresponding results in \cite{Kam18,Mes06,Mes14,SM10}. 
We remark that establishing this fact for weighted popular common independent sets 
is no longer a straightforward extension of the previous arguments and requires a more involved analysis. 

A standard approach for proving the necessity of the well-formedness is 
to show that any solution violating well-formedness can be improved, 
in the sense of head-to-head elections, by exchanging several elements.
Both the existing literature and the present work adopt this approach.
In matching problems and their extensions with matroid constraints, 
the exchanges that preserve feasibility can be identified relatively easily.
In contrast, in our setting, determining which exchanges maintain feasibility is more delicate.
To overcome this difficulty, we develop a nontrivial matroid-theoretic analysis that characterizes feasible exchanges 
and plays a central role in our proofs.


\paragraph{Refinement of the promotion paths.}
As mentioned above, the well-formedness is only 
a necessary condition for popularity. 
In order to obtain a necessary and sufficient condition, 
we 
generalize the notion of promotion paths 
to the setting of weighted popular common independent sets. 
Using these generalized promotion paths, we obtain a necessary and sufficient condition for popularity. 
However, this characterization is not well suited for algorithm design, 
because it is solution-dependent. 

To obtain a solution-independent characterization, 
a natural approach would be to generalize Mestre's pruning algorithm for weighted popular matchings~\cite{Mes06,Mes14}.
In the setting of weighted popular common independent sets, however, 
there is no underlying input graph on which such a pruning procedure can be performed, 
making it difficult to generalize the pruning-based approach.

To overcome this difficulty, we propose a completely different approach that does not rely on a pruning algorithm.
More specifically, we introduce the notion of \emph{essential promotion paths}, a constrained variant of promotion paths.
The development and analysis of this notion constitute the most technically intricate part of our work.
A key feature of essential promotion paths is that their existence is independent of the choice of solution under a mild assumption,
whereas the existence of promotion paths generally depends on the choice of solution.
By exploiting essential promotion paths, we derive a solution-independent characterization of weighted popular common independent sets,
which forms the foundation of our polynomial-time algorithm.



\subsection{Organization of the Paper}
The remainder of this paper is organized as follows.
In Section~\ref{sec:preliminary}, we review basic notation and properties of matroids and modify the instance to simplify the presentation.
In Section~\ref{sec:wellformed}, we define the concept of well-formedness and prove that it is a necessary condition for popularity.
In Section~\ref{se:algorithm}, we present a solution-independent characterization of popular common independent sets and develop an algorithm for finding a popular common independent set, if one exists, thereby proving Theorem~\ref{thm:main}. 
Some proofs are deferred to Section~\ref{sec:proofofkeytheorems}.
In Section~\ref{sec:conclusion}, we conclude this paper with several remarks and open problems.

\section{Preliminaries}
\label{sec:preliminary}

For a positive integer $k$, let $[k]$ denote $\{1,2,\ldots ,k\}$.
For a set $X$ and an element $e$, 
we denote $X \cup \{e\}$ and $X \setminus \{e\}$ by $X + e$ and $X-e$, respectively. 

\subsection{Matroid}
\label{sec:matroid}
A {\em matroid} is a pair $\M = (E, \I)$ of a finite set $E$ and 
a nonempty family $\I \subseteq 2^E$ of subsets of $E$ such that 
\begin{itemize}
    \item if $X \subseteq Y$ and $Y \in \I$, then $X \in \I$; and 
    \item if $X\in \I$, $Y \in \I$, and $|Y| > |X|$, then there exists $e \in Y \setminus X$ such that $X + e \in \I$. 
\end{itemize}
The sets in $\I$ are referred to as {\em independent sets}. 

For a matroid $\M=(E, \I)$, we define the {\em rank function} ${\rm rank} \colon 2^E \to \mathbb{Z}_{\ge 0}$ 
and the {\em closure operator} $\cl \colon 2^E \to 2^E$ by 
\begin{align*}
{\rm rank} (X) &:= \max \{ |Y| \mid Y \subseteq X,\ Y \in \I\}, \\
\cl (X) &:= \{ e \in E \mid {\rm rank} (X +e) = {\rm rank} (X) \}
\end{align*}
for any $X \subseteq E$. 
The closure operator is sometimes called the {\em span function} in the literature. 

An inclusion-wise maximal independent set is called a {\em base} of $\M$.
We denote by $\B$ the family of bases of $\M$.
It is well known that the family of bases satisfies the {\em simultaneous exchange property}; see e.g.~\cite[Theorem~39.12]{Schrijver2003}: 
\begin{quote}
    for any two bases $B_1$ and $B_2$ of $\M$ and any element $x \in B_1 \setminus B_2$, 
there exists an element $y \in B_2 \setminus B_1$ 
such that both $B_1 - x + y$ and $B_2 + x - y$ are bases of $\M$.
\end{quote} 
The following results on bases are well known and will be used later in this paper.

\begin{lemma}[\mbox{see~\cite[Theorem 39.13]{Schrijver2003}}]
\label{lem:uniquematching}
    Let $B$ be a base of a matroid $\M=(E, \I)$, let $X \subseteq B$, and let $Y \subseteq E \setminus B$. 
    If the bipartite graph $G=(X, Y; E')$ defined by $E':= \{ \{x, y\} \mid x \in X,\ y \in Y,\ B - x + y \in \B \}$ 
    contains a unique perfect matching, then $(B \setminus X) \cup Y$ is also a base of $\M$. 
\end{lemma}

An inclusion-wise minimal non-independent set is called a {\em circuit} of $\M$. 
For a base $B$ of $\M$ and an element $e \in E \setminus B$, 
the set $B + e$ contains a unique circuit, denoted by $C(B, e)$, which is called the \emph{fundamental circuit} of $B$ and $e$. 
Note that $e' \in B$ belongs to $C(B, e) - e$ if and only if $B + e - e' \in \I$. 
The following lemma is a straightforward property of fundamental circuits, but we state it explicitly because it will be used multiple times later.

\begin{lemma}\label{lem:clcircuit}
    Let $B$ be a base of a matroid $\M=(E, \I)$, $e \in E \setminus  B$, and $X\subseteq E$.
    Then, $e \in \cl(B\cap X)$ if and only if $C(B, e)- e \subseteq X$.
\end{lemma}

\begin{proof}
    For the ``only if'' part, 
    suppose that $e \in \cl(B\cap X)$, that is, $(B \cap X) + e \not\in \I$. 
    Since $C(B, e)$ is the unique minimal non-independent set in $B + e$, we obtain $C(B, e) \subseteq (B \cap X) + e$. 
    Therefore, $C(B, e) - e \subseteq X$ holds. 

    For the ``if'' part, 
    suppose that $C(B, e)- e \subseteq X$. 
    This together with $C(B, e)- e \subseteq B$ shows that $C(B, e) \subseteq (B \cap X) + e$, and hence $(B \cap X) + e \not\in \I$. 
    Therefore, $e \in \cl(B\cap X)$ holds. 
\end{proof}

For more details on matroids, we refer the reader to~\cite{oxley2011matroid,Schrijver2003}. 





\subsection{Instance with Last Resort}
\label{sec:lastresort}


To simplify the description, we modify the instance of the weighted popular common independent set problem so that  
every popular common independent set is a common base of the matroids. 
This is achieved by adding a new element $L(v)$ to $E_v$ for each $v \in V$, where $L(v)$ is the least preferred element of $v$ (called the \emph{last resort}), 
and extending $\M_{\mathrm{part}}$ and $\M$ to $\M'_{\mathrm{part}}=(E', \I'_{\mathrm{part}})$ and $\M'=(E', \I')$ as follows: 
\begin{align*}
E' &:= E \cup \{L(v) \mid v \in V\}, \\
\I'_{\mathrm{part}} &:= \{ X \subseteq E' \mid |X \cap (E_v + L(v))| \le 1 \text{ for all } v \in V \}, \\ 
\I' &:= \{X \subseteq E' \mid X \cap E \in \I,\ |X| \le |V|\}.
\end{align*}
It is not difficult to see that both $\M'_{\mathrm{part}}=(E', \I'_{\mathrm{part}})$ and $\M'=(E', \I')$ are matroids.
Moreover, this modification does not affect the problem, since a solution for $\M_{\mathrm{part}}$ and $\M$
corresponds to a solution for $\M'_{\mathrm{part}}$ and $\M'$, and vice versa, by adding or removing the last resorts. 

In what follows, the objects in the modified instance are simply denoted by $\M_{\mathrm{part}}$, $\M$, and $E_v$, etc.,  
and we suppose that each $E_v$ contains the last resort $L(v)$. 
The following lemma follows immediately from the definition of $L(v)$.


\begin{lemma}\label{ob:assumebase}
Assume that each $E_v$ contains the last resort $L(v)$.
\begin{enumerate}[label={\textup{(\arabic*)}}]
    \item Every popular common independent set of $\M_{\mathrm{part}}$ and $\M$ is a common base of $\M_{\mathrm{part}}$ and $\M$.
    \item A common base $B$ of $\M_{\mathrm{part}}$ and $\M$ is popular if and only if
    $\Delta_w(B, B') \geq 0$ for any common base $B'$ of $\M_{\mathrm{part}}$ and $\M$.
\end{enumerate}
\end{lemma}
\begin{proof}
    For a common independent set $I$,
    let $V(I) := \{ v \in V \mid I \cap E_v \neq \emptyset \}$.

    (1)
    Assume to the contrary that there exists a popular common independent set $I$ that is not a common base, i.e., $|I| \le |V|-1$.
    Then, the set $V \setminus V(I)$ is nonempty.
    In this case,
    for the common base $B(I) := I \cup \{ L(v) \mid v \in V \setminus V(I) \}$,
    we have $\Delta_w(I, B(I)) = - \sum_{v \in V \setminus V(I)} w(v) < 0$,
    which contradicts the popularity of $I$. 

    (2)
    The only-if part is clear.
    For the if part, 
    suppose that $\Delta_w(B, B') \geq 0$ holds for any common base $B'$
    and take an arbitrary common independent set $I$.
    It suffices to see $\Delta_w(B, I) \geq 0$.

    As in the above argument, we construct the common base $B(I) = I \cup \{ L(v) \mid v \in V \setminus V(I) \}$.
    Then, we obtain
    \begin{align*}
        \Delta_w(B, I)
        &= \sum_{v \in V:\, v \text{ prefers } B \text{ to } I} w(v)
        - \sum_{v \in V:\, v \text{ prefers } I \text{ to } B} w(v)\\
        &= \sum_{v \in V(I):\, v \text{ prefers } B \text{ to } B(I)} w(v) + \sum_{v \in V \setminus V(I)} w(v) - \sum_{v \in V:\, v \text{ prefers } B(I) \text{ to } B} w(v)\\
        &\ge \Delta_w(B, B(I)) \ge 0,
    \end{align*}
    which completes the proof.
\end{proof}

With this lemma, we mainly focus on a common base of $\M_{\mathrm{part}}$ and $\M$.
For any common base $B$ of $\M_{\mathrm{part}}$ and $\M$, 
and any $v\in V$,    
the set $B \cap E_v$ contains exactly one element, which we denote by $B(v)$. 
That is, $B(v)$ is the element assigned to $v$ in $B$.



Without loss of generality, we may assume that $\{e\} \in \I$ for every $e \in E$, i.e., $\M$ has no loops, 
since otherwise we can remove $e$ without affecting the problem.

Throughout the paper, if no confusion may arise, a common base (resp. common independent set) of $\M_{\mathrm{part}}$ and $\M$ is simply called a common base (resp. common independent set). 
We use ${\rm rank}(X)$, $\cl(X)$, $C(B, e)$, and $\B$ to denote 
the rank, the closure, the fundamental circuit, and the family of bases with respect to the matroid $\M$, respectively.

\section{Well-formed Common Base}
\label{sec:wellformed}

The objective of this section is to introduce the concept of well-formedness and to show that it is a necessary condition for a common base to be popular; see Proposition~\ref{prop:populartowellformed} below.

Let $w_1 > w_2 > \cdots > w_p$ be the distinct weights assigned to applicants of $V$, i.e., 
\[
\{w(v) \mid v \in V\} = \{w_1, w_2, \ldots, w_p\}. 
\]
For each $i \in [p]$, denote 
\[
V_i = \{v \in V \mid w(v) = w_i\}. 
\]
We also denote $V_0 = V_{\le 0} = \emptyset$ and $V_{\le i} = V_1 \cup \cdots \cup V_i$ for each $i \in [p]$.

For each $v \in V$, we inductively define its \emph{first element} in $E_v$ in the following way.
Let $F_0 = F_{\le 0} = \emptyset$. 
For each $i \in [p]$ and each $v\in V_i$, assuming that $F_{\le i-1} \subseteq E$ has been defined, 
we define the first element $f(v) \in E_v$ to be the most preferred element with respect to $\succ_v$ in $E_v \setminus \cl(F_{\le i-1})$. 
We then set 
\begin{align*}
F_i := \{f(v) \mid v \in V_i\}, \qquad F_{\le i} := F_1 \cup \cdots \cup F_i, \qquad F := F_{\leq p} = \{ f(v) \mid v \in V \}.
\end{align*}
Note that the first element $f(v)$ is well defined for every $v \in V$, since $L(v) \in E_v \setminus \cl(F_{\le i-1})$ for each $i \in [p]$ and each $v \in V_i$, where $L(v)$ is the last resort of $v$.


Next, for each $i \in [p]$ and each $v \in V_i$, 
we define the \emph{second element} $s(v) \in E_v$ to be the most preferred element with respect to $\succ_v$ in $E_v \setminus \cl(F_{\le i})$.
If $E_v \setminus \cl(F_{\le i}) = \emptyset$, then $v$ has no second element, and we denote $s(v) = \bot$. 
Let $S$ be the set of second elements, i.e.,
\begin{align*}
    S := \{ s(v) \mid v \in V,\ s(v) \neq \bot \}.
\end{align*}


\begin{figure}[!t]
    \centering
    \begin{subfigure}{0.48\linewidth}
        \centering
        \includegraphics[width=0.5\linewidth]{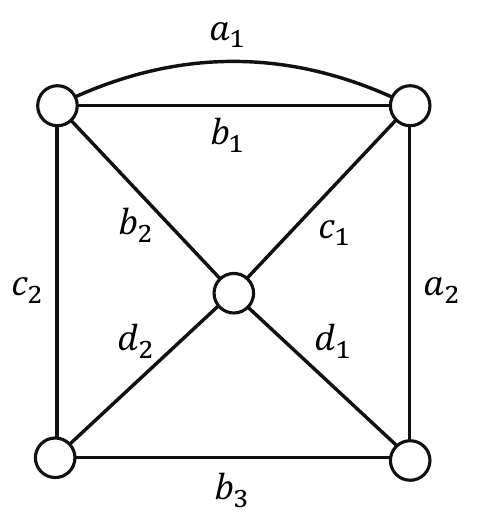}
        \caption{The graph defining a graphic matroid.}
        \label{fig:sub1}
    \end{subfigure}
    \hspace{0.01\linewidth}
    \begin{subfigure}{0.48\linewidth}
        \centering
        \includegraphics[width=0.5\linewidth]{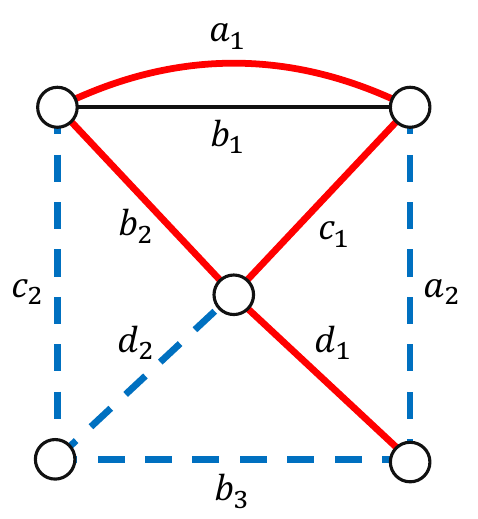}
        \caption{Thick red edges and blue dashed edges represent the first and second elements, respectively.}
        \label{fig:sub2}
    \end{subfigure}

    \vspace{0.1em}

    \begin{subfigure}{0.48\linewidth}
        \centering
        \includegraphics[width=0.5\linewidth]{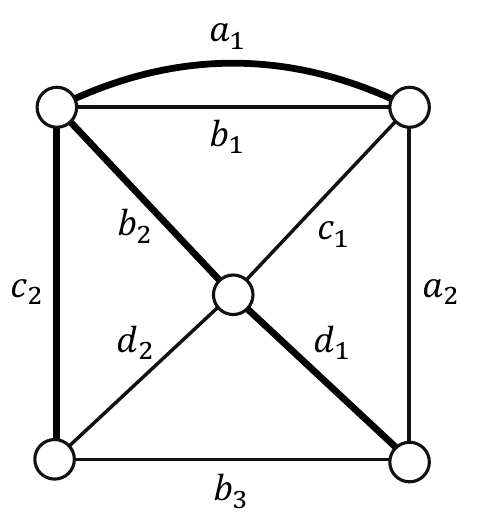}
        \caption{A well-formed common base that is not popular.}
        \label{fig:sub3}
    \end{subfigure}
    \hspace{0.01\linewidth}
    \begin{subfigure}{0.48\linewidth}
        \centering
        \includegraphics[width=0.5\linewidth]{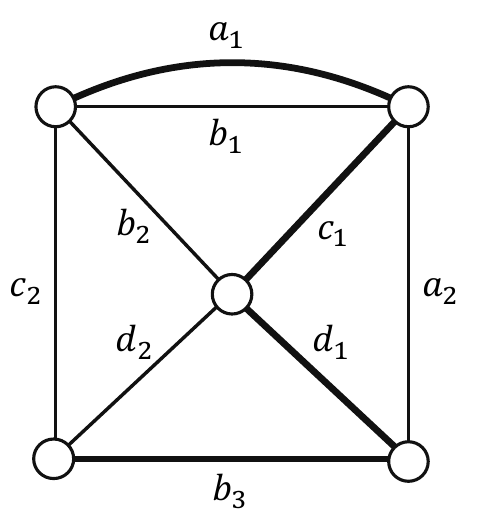}
        \caption{A popular common base.}
        \label{fig:sub4}
    \end{subfigure}

    \caption{An example illustrating well-formedness and popularity.}
    \label{fig:1}
\end{figure}

\begin{definition}\label{def:wellformed}
    A common base $B$ is \emph{well-formed} if the following conditions hold:
    \begin{itemize}
        \item $B \subseteq F \cup S$, i.e., $B(v) \in \{f(v), s(v)\}$ for all $v\in V$; and
        \item $F_{\le i} \subseteq \cl(B\cap F_{\le i})$ for all $i\in [p]$.
        \qedhere
    \end{itemize}
\end{definition}

\begin{remark}\label{rem:clequivalent}
For each $i \in [p]$, the condition $F_{\le i} \subseteq \cl(B\cap F_{\le i})$ is equivalent to $\cl(F_{\le i}) = \cl(B\cap F_{\le i})$. 
It is also equivalent to $|B \cap F_{\le i}| = {\rm rank} (F_{\le i})$; see Lemma~\ref{lem:remark42}. 
\end{remark}
The second condition of well-formedness can be interpreted as follows:
first, $B$ is an independent set of $\M$ that contains as many elements as possible from $F_1$;
next, subject to this condition, it contains as many elements as possible from $F_2$;
then, subject to these two conditions, it contains as many elements as possible from $F_3$;
and so on.
\begin{example}
Figure~\ref{fig:1} gives an example illustrating the notions of well-formedness and popularity.
There are four applicants $a,b,c,d$, and the partition classes are
$E_a=\{a_1,a_2,L(a)\}$, $E_b=\{b_1,b_2,b_3,L(b)\}$, $E_c=\{c_1,c_2,L(c)\}$, and $E_d=\{d_1,d_2,L(d)\}$.
The matroid $\M$ is obtained from the graphic matroid of the graph shown in Figure~\ref{fig:sub1} by adding the last-resort elements.
The weights of the applicants are $w(a)=3$, $w(b)=w(c)=2$, and $w(d)=1$.

The preferences of the applicants are given in the following table:
$$
\begin{array}{c|c}
a & a_1 \succ_a a_2 \succ_a L(a) \\
b & b_1 \succ_b b_2 \succ_b b_3 \succ_b L(b) \\
c & c_1 \succ_c c_2 \succ_c L(c) \\
d & d_1 \succ_d d_2 \succ_d L(d)
\end{array}
$$

Then, the first elements are $f(a)=a_1$, $f(b)=b_2$, $f(c)=c_1$, and $f(d)=d_1$, which are shown as thick red edges in Figure~\ref{fig:sub2}.
The second elements are $s(a)=a_2$, $s(b)=b_3$, $s(c)=c_2$, and $s(d)=d_2$, which are shown as blue dashed edges in Figure~\ref{fig:sub2}.

First, consider a common base $B=\{a_1,b_2,c_2,d_1\}$, which is represented by the thick black edges in Figure~\ref{fig:sub3}.
This set $B$ is well-formed.
However, $B$ is not popular.
Indeed, let $B'=\{L(a),b_1,c_1,d_1\}$.
Then $B'$ is also a common base, and comparing $B$ with $B'$ gives
$$
\Delta_w(B,B')=w(a)-w(b)-w(c)=3-2-2=-1<0.
$$
Thus, $B$ is not popular.
On the other hand, $B^*=\{a_1,b_3,c_1,d_1\}$ is a popular common base in this instance, and
it is represented by the thick black edges in Figure~\ref{fig:sub4}.
\end{example}


In what follows in this section, we prove the following proposition. 

\begin{proposition}\label{prop:populartowellformed}
Every popular common base is well-formed.
\end{proposition}

\subsection{Basic Lemmas}

Before proving Proposition~\ref{prop:populartowellformed}, we prove the following two lemmas.

    

\begin{lemma}\label{lem:f(v)>B(v)}
    Let $i\in [p]$ and $v\in V_i$.
    Then the following statements hold for any common base $B$. 
    \begin{enumerate}[label={\textup{(\arabic*)}}]
         \item If $\cl(F_{\le {i-1}}) = \cl(B\cap F_{\le {i-1}})$ and $B(v)\neq f(v)$, then $f(v) \succ_v B(v)$. 
        \item If $\cl(F_{\le {i}}) = \cl(B\cap F_{\le {i}})$ and $B(v)\notin \{f(v), s(v)\}$, then $s(v) \succ_v B(v)$ (implying that $s(v) \neq \bot$).
    \end{enumerate}

\end{lemma}
\begin{proof}
First note that $B(v) \notin \cl(B-B(v))$ follows from the fact that $B$ is a base of $\M$ and $B(v)\in B$. 

(1) Suppose that $\cl(F_{\le {i-1}}) = \cl(B\cap F_{\le {i-1}})$ and $B(v)\neq f(v)$. 
Since $v\in V_i$, we have $B\cap F_{\le {i-1}} \subseteq B- B(v)$, 
and 
hence $\cl(B\cap F_{\le {i-1}}) \subseteq \cl(B- B(v))$. 
It then follows from $B(v) \notin \cl(B-B(v))$  that 
\[
B(v)\notin \cl(B\cap F_{\le {i-1}}) = \cl(F_{\le {i-1}}).
\]
Then $f(v) \succ_v B(v)$ follows from the definition of $f(v)$ and the assumption $B(v)\neq f(v)$. 

(2) Suppose that $\cl(F_{\le {i}}) = \cl(B\cap F_{\le {i}})$ and $B(v)\notin \{f(v), s(v)\}$. 
It follows from $B(v)\neq f(v)$ that $B(v)\in B \setminus F_{\le i}$, in particular $B\cap F_{\le {i}} \subseteq B - B(v)$, 
and hence $\cl(B\cap F_{\le {i}}) \subseteq \cl(B- B(v))$. 
From $B(v)\notin \cl(B- B(v))$, 
we obtain that 
\[
B(v)\notin \cl(B\cap F_{\le {i}}) = \cl(F_{\le {i}}).
\]
Then $s(v) \succ_v B(v)$ follows from the definition of $s(v)$ and the assumption $B(v)\neq s(v)$. 
\end{proof}

\begin{lemma}\label{lem:circuit-level}
    Suppose that $B$ is a popular common base, 
    and that $v \in V_i$.
    Let $e \in E_v$ with $e \succ_v B(v)$.
    Then,
    each applicant $u \in V$ with $B(u) \in C(B, e) - e$ satisfies
    $w(u) \geq w(v)$ (or equivalently, $u \in V_j$ with $j \leq i$)
    and
    $u \neq v$.
\end{lemma}
\begin{proof}
We note that $C(B, e)$ is well defined in the statement. 
Indeed, since $e$ and $B(v)$ are distinct elements in $E_v$, 
we have $e \notin B$. 

We first prove that $B(u) \in C(B, e) - e$ implies $u\neq v$. 
Suppose to the contrary that $B(v) \in C(B, e) - e$. 
Since both $B(v)$ and $e$ belong to $E_v$,
the set $B' := B - B(v) + e$ forms a common base.
By the assumption that $e \succ_v B(v)$,
we have $\Delta_w(B,B') = -w(v) < 0$,
which contradicts that $B$ is popular.

We next prove that $w(u) \geq w(v)$
for any $u \in V-v$ with $B(u) \in C(B, e)-e$.
We see that $I := B - B(u) + e - B(v)$ 
is an independent set of $\M_{\mathrm{part}}$, because 
$e$ and $B(v)$ belong to $E_v$. 
Furthermore, $B(u) \in C(B, e)-e$ implies that $B - B(u) + e$ is an independent set of $\M$, and hence, so is $I$.
Therefore, $I$ is a common independent set.
By the assumption that $e \succ_v B(v)$,
we have $\Delta_w(B,I) = -w(v) + w(u)$. 
Since $B$ is popular, this value is non-negative, implying that $w(u) \geq w(v)$. 
%
%
%
%
\end{proof}

\subsection{Proof of Proposition~\ref{prop:populartowellformed}}


To prove Proposition~\ref{prop:populartowellformed}, we show that every popular common base satisfies the two conditions in Definition~\ref{def:wellformed}.

\begin{lemma}\label{lem:span}
    If $B$ is a popular common base, then $F_{\le i} \subseteq \cl(B\cap F_{\le i})$ for all $i\in [p]$.
\end{lemma}
\begin{proof}
Let $B$ be a popular common base. 
We prove the claim by induction on $i$. 
For the base case $i=0$, since $F_{\le 0} = \emptyset$, the claim obviously holds. 
For the induction step, let $i\in [p]$ and suppose that $F_{\le j} \subseteq \cl(B\cap F_{\le j})$ holds for all $j<i$. 
As mentioned in Remark~\ref{rem:clequivalent}, this is equivalent to
\begin{equation}
    \cl(F_{\le j}) = \cl(B\cap F_{\le j}) \text{ for all } j<i. \label{eq:sameclosure} 
\end{equation} 
We prove that the claim holds for $i$.

Suppose, for contradiction, that $F_{\le i} \nsubseteq \cl(B\cap F_{\le i})$, i.e., $F_{\le i} \setminus \cl(B\cap F_{\le i}) \neq \emptyset$.
Our aim in the following is to construct a common independent set $I$
satisfying $\Delta_w(B,I) < 0$;
this contradicts the popularity of $B$
and completes the proof.

Let $\hat{V}_i$ be the nonempty subset of $V$ defined as
\begin{align*}
    \hat{V}_i :=&\; \{ v \in V \mid f(v) \in F_{\leq i} \setminus \cl(B\cap F_{\le i}) \}\\
    =&\; \{ v \in V \mid f(v) \in F_{i} \setminus \cl(B\cap F_{\le i}) \},
\end{align*}
where the second equality follows from $F_{\le i} = F_{\le i-1} \cup F_i$ and $F_{\le i-1}\subseteq \cl(B\cap F_{\le i-1})\subseteq \cl(B\cap F_{\le i})$ (by the induction hypothesis). 
Note that this equality implies $\hat{V}_i \subseteq V_i$.
%
%
%
Let $v \in \hat{V}_i$. 
Since $v \in \hat{V}_i$ implies $f(v) \not\in B$, i.e., $B(v) \neq f(v)$, 
we have that 
\begin{equation}
    f(v) \succ_v B(v) \label{eq:preferf(v)}
\end{equation} 
by (\ref{eq:sameclosure}) and Lemma~\ref{lem:f(v)>B(v)} (1). 

By Lemma~\ref{lem:clcircuit} with $e = f(v)$ and $X = F_{\le i}$, 
we obtain from $f(v) \not\in \cl(B \cap F_{\le i})$ that 
$(C(B, f(v))-f(v))\setminus F_{\le i} \neq \emptyset$. 
Since $(C(B, f(v))-f(v))\setminus F_{\le i} \subseteq B$, 
there exists an applicant $v' \in V_{i'}$ such that 
\begin{equation}
    B(v') \in (C(B, f(v))-f(v))\setminus F_{\le i}. \label{eq:property1v'}
\end{equation}
Then, we obtain $v' \neq v$ and $i' \leq i$ from (\ref{eq:preferf(v)}), (\ref{eq:property1v'}), and Lemma~\ref{lem:circuit-level}.
Moreover, since $B(v') \notin F_{\leq i}$, particularly, $B(v') \notin F_{i'}$,
we obtain $B(v') \neq f(v')$.
Thus, 
\begin{equation}
f(v') \succ_{v'} B(v') \label{eq:preferf(v')}
\end{equation}
holds by (\ref{eq:sameclosure}) and Lemma~\ref{lem:f(v)>B(v)} (1).


We now divide the analysis into two cases, depending on whether
\begin{equation}
    f(v') \in \cl(B \cap F_{\leq i}) \label{eq:property2v'}
\end{equation}
holds. 


\medskip

\noindent
\textbf{Case 1.} 
%
%
%
%
%
Suppose that (\ref{eq:property2v'}) holds. 
Since $f(v') \not\in B$ and $\mathcal{M}$ is assumed to have no loops, $C(B, f(v')) - f(v')$ is a well-defined nonempty subset of $B$. 
Take an applicant $v'' \in V_{i''}$ such that 
\begin{equation}
    B(v'') \in C(B, f(v')) - f(v'), \label{eq:f(v'')}
\end{equation} 
where $i''$ is chosen to be maximized. 
By Lemma~\ref{lem:clcircuit} with $e = f(v')$ and $X = F_{\le i}$, we obtain from (\ref{eq:property2v'}) that $C(B, f(v'))-f(v')\subseteq F_{\le i}$, and hence 
$B(v'') \in F_{\le i}$, i.e., $B(v'') = f(v'')$.


We now claim that $i'' \ge i'$; 
otherwise, 
we obtain $C(B, f(v'))-f(v')\subseteq F_{\le i'-1}$ by the maximality of $i''$, 
implying that $f(v') \in \cl(C(B, f(v'))-f(v'))\subseteq \cl(F_{\le i'-1})$; this contradicts the definition of $f(v')$.
Furthermore, we obtain $v'' \neq v'$ and $i'' \leq i'$ from (\ref{eq:preferf(v')}), (\ref{eq:f(v'')}), and Lemma~\ref{lem:circuit-level}. 
Therefore, $i'' = i'$ holds, that is, $w(v')=w(v'')$.
Since $B(v'') = f(v'')$ and $B(v) \neq f(v)$,
we obtain $v'' \neq v$.
Thus, $v, v'$, and $v''$ are distinct.

Let $I := B - B(v) + f(v) - B(v') + f(v') - B(v'')$.
We claim that $I$ is a common independent set.
Indeed, it is clear that $I \in \I_{\mathrm{part}}$.
In addition, since 
\begin{itemize}
    \item $B + f(v) - B(v') \in \I$ by (\ref{eq:property1v'}), 
    \item $B + f(v') - B(v'') \in \I$ by (\ref{eq:f(v'')}), and
    \item $B + f(v') - B(v') \not\in \I$ by (\ref{eq:preferf(v')}) and the popularity of $B$,  
\end{itemize}
the set $B' := B + f(v) - B(v') + f(v') - B(v'')$ is a base of $\M$ by Lemma~\ref{lem:uniquematching}.
Thus, $I = B' - B(v)$ is an independent set of $\M$. 

By (\ref{eq:preferf(v)}), (\ref{eq:preferf(v')}), and $w(v')=w(v'')$, 
we have that 
\[
\Delta_w(B,I) = -w(v) -w(v')+ w(v'') = -w(v) < 0, 
\]
which contradicts the popularity of $B$.


\medskip

\noindent
\textbf{Case 2.} 
Suppose that (\ref{eq:property2v'}) does not hold. 
In this case, $f(v') \not\in \cl(B \cap F_{\le i})$ implies that 
$f(v') \not\in \cl(B \cap F_{\le i-1}) = \cl(F_{\le i-1})$ by (\ref{eq:sameclosure}), 
and hence $i' \ge i$. 
This together with $i' \le i$ shows that $i' = i$. 

Since (\ref{eq:property2v'}) does not hold, by applying Lemma~\ref{lem:clcircuit} with $e = f(v')$ and $X = F_{\le i}$, 
we have that $(C(B, f(v'))-f(v'))\setminus F_{\le i} \neq \emptyset$. 
Let $v'' \in V_{i''}$ be an applicant such that 
\begin{equation}
    B(v'') \in (C(B, f(v'))-f(v'))\setminus F_{\le i}. \label{eq:property1v''}
\end{equation}
Then, we obtain $v'' \neq v'$ and $i'' \leq i'$ from (\ref{eq:preferf(v')}), (\ref{eq:property1v''}), and Lemma~\ref{lem:circuit-level}.
Moreover, since $B(v'') \notin F_{\leq i}$, particularly, $B(v'') \notin F_{i''}$,
we obtain $B(v'') \neq f(v'')$.

If $v'' = v$, then $B':= B - B(v) + f(v) - B(v') + f(v')$ is a common base.
Indeed, it is clear that $B'$ is a base of $\M_{\mathrm{part}}$.
In addition, since 
\begin{itemize}
    \item $B + f(v) - B(v') \in \I$ by (\ref{eq:property1v'}), 
    \item $B + f(v') - B(v) \in \I$ by (\ref{eq:property1v''}) and $B(v) = B(v'')$, and
    \item $B + f(v') - B(v') \not\in \I$ by (\ref{eq:preferf(v')}) and the popularity of $B$,  
\end{itemize}
$B'$ is a base of $\M$ by Lemma~\ref{lem:uniquematching}.
Since $\Delta_w(B,B') = -w(v) -w(v') < 0$ by (\ref{eq:preferf(v)}) and (\ref{eq:preferf(v')}), 
this contradicts the popularity of $B$. 

Therefore, $v''\neq v$ holds, and hence, $v$, $v'$, and $v''$ are distinct applicants. 
Then we derive a contradiction by considering the following two cases separately. 
\begin{itemize}
    \item 
    Suppose that $f(v'') \in \cl(B \cap F_{\leq i})$. 
    In this case, since $v' \in \hat V_i$, we can derive a contradiction by applying the same argument as Case 1, 
    in which $v$ and $v'$ are replaced by $v'$ and $v''$, respectively. 
    \item 
    Suppose that $f(v'') \not\in \cl(B \cap F_{\leq i})$. 
    In this case, $f(v'') \not\in \cl(B \cap F_{\le i-1}) = \cl(F_{\le i-1})$ by (\ref{eq:sameclosure}), 
    and hence $i'' \ge i=i'$. 
    This together with $i'' \le i'$ shows that $i'' = i' = i$. 
    Let $I:= B - B(v) + f(v) - B(v') + f(v') - B(v'')$, which is a common independent set, 
    as shown in the same way as in Case 1.
    By (\ref{eq:preferf(v)}), (\ref{eq:preferf(v')}), and $i'' = i' = i$, we have 
    \[
    \Delta_w(B,I) = -w(v) -w(v')+ w(v'') = -w(v) < 0, 
    \]
    which contradicts the popularity of $B$.
\end{itemize}

This completes the proof. 
%
%
\end{proof}

As shown below, the first condition in Definition~\ref{def:wellformed} can also be proved by an argument similar to that used in Lemma~\ref{lem:span}.
\begin{lemma}\label{lem:fors}
    If $B$ is a popular common base, then $B(v) \in \{f(v), s(v)\}$ for all $v\in V$.
\end{lemma}

\begin{proof}
Let $B$ be a popular common base. 
Fix any $i \in [p]$.
Suppose, for contradiction, that $B(v) \notin \{f(v), s(v)\}$ for some $v \in V_i$.
It suffices to show that there is a common independent set $I$
satisfying $\Delta_w(B, I) < 0$,
which contradicts the popularity of $B$.

By Lemma~\ref{lem:span} and Remark~\ref{rem:clequivalent}, $\cl(F_{\le j}) = \cl(B\cap F_{\le j})$ holds for all $j\in [p]$.
Thus,
by Lemma~\ref{lem:f(v)>B(v)} (2), we have 
\begin{equation}
    s(v) \succ_v B(v). \label{eq:s(v)01}
\end{equation}
If $C(B, s(v))-s(v)\subseteq F_{\le i}$, then $s(v) \in \cl(C(B, s(v))-s(v))\subseteq \cl(F_{\le i})$, which contradicts the choice of $s(v)$.
Thus, $(C(B, s(v))-s(v)) \setminus F_{\le i} \neq \emptyset$ holds,
i.e.,
there exists $v'\in V_{i'}$ such that 
\begin{equation}
    B(v') \in (C(B, s(v))-s(v))\setminus F_{\le i}. \label{eq:s(v)02}
\end{equation}
By (\ref{eq:s(v)01}), (\ref{eq:s(v)02}), and Lemma~\ref{lem:circuit-level},
we obtain $v' \neq v$ and $i' \leq i$.
Moreover, since $B(v') \notin F_{\leq i}$, particularly, $B(v') \notin F_{i'}$,
we obtain $B(v') \neq f(v')$.
Thus, we obtain 
\begin{equation}
    f(v') \succ_{v'} B(v') \label{eq:s(v)03}
\end{equation} 
by Lemma~\ref{lem:f(v)>B(v)} (1).
Since $f(v') \in F_{\leq i'} \subseteq \cl(F_{\leq i'}) = \cl(B \cap F_{\leq i'})$,
we have $C(B, f(v'))-f(v') \subseteq F_{\le i'}$ by Lemma~\ref{lem:clcircuit}.
Furthermore, $C(B, f(v'))-f(v') \not\subseteq F_{\le i'-1}$; otherwise $f(v') \in \cl(C(B, f(v'))-f(v'))\subseteq \cl(F_{\le i'-1})$, which contradicts the choice of $f(v')$. 
Thus, there exists $v''\in V_{i'}$ such that 
\begin{equation}
    f(v'') \in C(B, f(v'))-f(v'). \label{eq:s(v)04}
\end{equation}
Note that this implies $B(v'') = f(v'')$, because $C(B, f(v'))-f(v') \subseteq B$. 
It follows from $B(v'') = f(v'')$ and $B(v) \neq f(v)$ that $v'' \neq v$.
We also obtain $v'' \neq v'$ from $f(v'') \neq f(v')$. 
Therefore, $v$, $v'$, and $v''$ are distinct applicants. 


Let $I := B - B(v) + s(v) - B(v') + f(v') - B(v'')$.
It is clear that $I \in \I_{\rm part}$. 
Since 
\begin{itemize}
    \item $B + s(v) - B(v') \in \I$ by (\ref{eq:s(v)02}), 
    \item $B + f(v') - B(v'') \in \I$ by (\ref{eq:s(v)04}) and $B(v'') = f(v'')$, and
    \item $B + f(v') - B(v') \not\in \I$ by (\ref{eq:s(v)03}) and the popularity of $B$,  
\end{itemize}
the set $B' := B + s(v) - B(v') + f(v') - B(v'')$ is a base of $\M$ by Lemma~\ref{lem:uniquematching}.
Thus, $I = B' - B(v)$ is an independent set of $\M$, 
implying that $I$ is a common independent set.
By (\ref{eq:s(v)01}), (\ref{eq:s(v)03}), and $w(v') = w_{i'} = w(v'')$, we obtain 
\[
\Delta_w(B,I) = -w(v) -w(v')+ w(v'') = -w(v) < 0, 
\]
which contradicts the popularity of $B$.
\end{proof}


Proposition \ref{prop:populartowellformed} follows from Lemmas~\ref{lem:span} and \ref{lem:fors}.

\section{Algorithm}
\label{se:algorithm}

We have already shown in Proposition~\ref{prop:populartowellformed} that every popular common base is well-formed.
In this section, we present a solution-independent characterization of when a well-formed common base is popular. 
Using this characterization, we develop an algorithm for finding a popular common base, if one exists, and thereby prove Theorem~\ref{thm:main}. 
Some proofs are deferred to Section~\ref{sec:proofofkeytheorems}.

For a subset $X \subseteq E$,
we say that a common base $B$ 
is \emph{$X$-including} (resp.~\emph{$X$-excluding})
if $X \subseteq B$ (resp.~$X \cap B = \emptyset$).
If $X = \{x\}$,
we simply write $x$-including and $x$-excluding instead of $\{x\}$-including and $\{x\}$-excluding,
respectively.

\subsection{Technical Preliminaries}

Let $B$ be a common base.
For $e \in E$, we define $w_B(e)$ as
\[
w_B(e) := 
\begin{cases}
0 & \mbox{if $e \in B$}, \\
w_i & \mbox{if $e \succ_v B(v)$ for some $v \in V_i$}, \\
-w_i & \mbox{if $e \prec_v B(v)$ for some $v \in V_i$}. 
\end{cases}
\]
Define an auxiliary bipartite digraph $D(B)$ on the vertex set $E$ with bipartition $(B, E \setminus B)$ as follows:  
for $x \in B$ and $y \in E \setminus B$, the digraph $D(B)$ contains an arc $(x, y)$ if $B - x + y \in \I_{\rm part}$,
i.e., $x$ and $y$ belong to the same partition class,
and it contains an arc $(y, x)$ if $B - x + y \in \I$,
i.e.,
$x \in C(B,y) - y$.
For a dicycle $C$ in $D(B)$, its weight is defined as $\sum_{e \in V(C)} w_B(e)$, 
where $V(C)$ denotes the set of vertices (elements) in $C$. 

\begin{lemma}\label{lem:negativecycle}
A common base $B$ is popular if and only if $D(B)$ contains no dicycle of positive weight with respect to $w_B$. 
\end{lemma}
\begin{proof}
For a common base $B'$,
we see that 
\begin{align*}
    \Delta_w(B, B') = \sum_{e \in B} w_B(e) - \sum_{e \in B'} w_B(e). 
\end{align*}
Therefore, $B$ is popular if and only if $B$ is a maximum weight common base with respect to $w_B$ by \Cref{ob:assumebase}~(2).
Furthermore, it is well known that $B$ is a common base of maximum weight if and only if $D(B)$ contains no dicycle of positive weight; see e.g.,~\cite[Theorem 41.5]{Schrijver2003}. This completes the proof.  
\end{proof}

A well-formed common base $B$
has intriguing properties as follows, which will be used several times later. 
\begin{lemma}\label{lem:lastarc}
Let $y \in \cl(F)$ and let $i^* \in [p]$ be the minimum index such that $y \in \cl (F_{\le i^*})$.
For any $y$-excluding well-formed common base $B$, we have
$C(B, y) - y \subseteq F_{\leq i^*}$ and $(C(B, y) - y) \cap F_{i^*} \neq \emptyset$.
\end{lemma}
\begin{proof}
We have $y \in \cl (F_{\le i^*})$ and $y \not\in \cl (F_{\le i^*-1})$
by the minimality of $i^*$.
Recall that ${\rm cl}(F_{\le i^*}) = {\rm cl}(B \cap F_{\le i^*})$
and ${\rm cl}(F_{\le i^*-1}) = {\rm cl}(B \cap F_{\le i^*-1})$ as shown in Remark~\ref{rem:clequivalent}.

Since $y \in {\rm cl}(F_{\le i^*}) = {\rm cl}(B \cap F_{\le i^*})$, we obtain
$C(B, y)- y \subseteq F_{\le i^*}$ by Lemma~\ref{lem:clcircuit}, that is, the first property holds. 
Furthermore, since $y \not\in {\rm cl}(F_{\le i^*-1}) = {\rm cl}(B \cap F_{\le i^*-1})$, 
we obtain $C(B, y)- y \not\subseteq F_{\le i^*-1}$ by Lemma~\ref{lem:clcircuit}. 
This together with $C(B, y)- y \subseteq F_{\le i^*}$ shows that $(C(B, y)- y) \cap F_{i^*} \neq \emptyset$.
%
\end{proof}

As special cases of this lemma, we obtain the following corollary. 

\begin{corollary}\label{cor:<}
For any $v \in V_i$ and $y \in E_v$, we have the following. 
\begin{enumerate}[label={\textup{(\arabic*)}}]
        \item 
        If $s(v) \prec_v y$,
        then $y \in \cl(F_{\leq i})$ and $C(B,y) - y \subseteq F_{\leq i}$ for any $y$-excluding well-formed common base $B$.
        \item 
        If $f(v) \prec_v y$,
        then $y\in \cl(F_{\leq i-1})$ and $C(B,y) - y \subseteq F_{\leq i-1}$ for any well-formed common base $B$.
        \item 
        If $y = f(v)$,
        then, for any $y$-excluding well-formed common base $B$, we have $C(B,y) - y \subseteq F_{\leq i}$ and $(C(B,y) - y) \cap F_{i} \neq \emptyset$.
    \end{enumerate}
\end{corollary}
\begin{proof}
If $s(v) \prec_v y$, then $y \in \cl(F_{\leq i})$ by the definition of $s(v)$. Therefore, the assertion~(1) follows from~\Cref{lem:lastarc}. 

If $f(v) \prec_v y$, then $y \in \cl(F_{\leq i-1})$ by the definition of $f(v)$.
Furthermore, since $y \notin \{s(v), f(v)\}$, any well-formed common base is $y$-excluding.
Therefore, the assertion~(2) follows from~\Cref{lem:lastarc}. 

If $y = f(v)$, then we have $y \in F_i \subseteq \cl(F_{\leq i})$ and $y \notin \cl(F_{\leq i-1})$.
Therefore, the assertion~(3) follows from~\Cref{lem:lastarc}.
%
\end{proof}


\subsection{(Essential) Promotion Path}
As shown in~\Cref{lem:negativecycle}, the popularity of a well-formed common base $B$ can be characterized by using a dicycle in $D(B)$. In this subsection, we characterize the popularity by using a dipath in $D(B)$ with some conditions. 

Let $B$ be a well-formed common base.
We first define a {\em promotion path} in $D(B)$ as follows. 
\begin{definition}
A dipath $P=(x_1, y_1, x_2, y_2, \dots , x_{\ell}, y_{\ell})$ in $D(B)$ is called a {\em promotion path (with respect to $B$)}
if $x_1 \in B$ and
$x_k \prec_{v_k} y_k$ for $k \in [\ell]$,
where $v_k \in V$ denotes the applicant with $x_k \in E_{v_k}$
for $k \in [\ell]$.
\end{definition}
Note that $x_k \in E_{v_k} \cap B$ and $y_k \in E_{v_k} \setminus B$ for $k \in [\ell]$, 
because $D(B)$ is bipartite between $B$ and $E\setminus B$ and $x_1\in B$.
In particular, we refer to
a promotion path that starts at $x$ as a \emph{promotion $x$-path}.

The following is an easy consequence of \Cref{cor:<}.
\begin{corollary}\label{cor:improving-dipath}
Let $P=(x_1, y_1, x_2, y_2, \dots , x_{\ell}, y_{\ell})$ be a promotion path with respect to $B$ with $\ell\ge2$,
in which $v_k \in V$ denotes the applicant with $x_k, y_k \in E_{v_k}$ for $k \in [\ell]$.
Then, we have $x_2, \dots, x_\ell \in F$
and $w(v_2) < \cdots < w(v_\ell)$.
In addition, $w(v_1) \leq w(v_2)$ if $x_1 = s(v_1)$,
and $w(v_1) < w(v_2)$ if $x_1 = f(v_1)$. 
\end{corollary}
\begin{proof}
Suppose that $v_k \in V_{i_k}$ for each $k \in [\ell]$.
Note that $x_1 \in \{f(v_1), s(v_1)\}$ since $B$ is well-formed.
Since $x_1 \prec_{v_1} y_1$,
by \Cref{cor:<}~(1) and~(2), 
we have $x_2 \in C(B, y_1) - y_1 \subseteq F_{\leq i_1}$ if $x_1 = s(v_1)$,
and $x_2 \in C(B, y_1) - y_1 \subseteq F_{\leq i_1-1}$ if $x_1 = f(v_1)$.
Thus, we obtain $x_2 \in F$. Moreover, $i_2 \leq i_1$ (or equivalently, $w(v_1) \leq w(v_2)$) if $x_1 = s(v_1)$,
and $i_2 \leq i_1-1$ (or equivalently, $w(v_1) < w(v_2)$) if $x_1 = f(v_1)$.
By repeatedly applying the same argument to $x_2, \dots , x_{\ell -1}$, 
we obtain $x_3, \dots, x_\ell \in F$ and $w(v_2) < w(v_3) <\cdots < w(v_\ell)$.
\end{proof}

We introduce $\alpha_B$ and $\beta_B$ as follows.
For $x \in B$,
we define
\begin{align}
    A_B(x) &:= \{ y \in F \setminus B \mid x \in C(B, y) - y\text{, i.e., } B -x + y \in \B \},\label{eq:A_B}\\
    \alpha_B(x) &:=
     \begin{cases}
        \max\{w(v) \mid v \in V,\ f(v) \in A_B(x) \} & \text{if $A_B(x) \neq \emptyset$},\\
        0 & \text{if $A_B(x) = \emptyset$}.
     \end{cases}\label{eq:alpha_B}
\end{align}
It is obvious that $\alpha_B(x) \geq 0$. 
Note also that if $\alpha_B(x) > 0$, i.e., if $A_B(x)$ contains some element $y$,
then $x \in F$, because $x \in C(B, y) - y \subseteq F$ by~\Cref{lem:lastarc}.
%
Similarly, for $y \in \cl(F) \setminus B$,
we define
\begin{align}\label{eq:beta_B}
    \beta_B(y) := \min \{ w(v) \mid v \in V,\ f(v) \in B,\ B + y - f(v) \in \B \}.
\end{align}
We note that $\beta_B(y) > 0$
for any $y \in \cl(F) \setminus B$, 
since
there exists $v \in V$ such that $f(v) \in C(B, y) - y$ by \Cref{lem:lastarc}.
Here, we use the fact that $C(B, y) - y \neq \emptyset$, which follows from our assumption that $\M$ has no loops. 


For a promotion path $P=(x_1, y_1, x_2, y_2, \dots , x_{\ell}, y_{\ell})$ with respect to a well-formed common base $B$,
the value $\alpha_B(x_k)$ is well defined for $k \in [\ell]$, because $x_k \in B$. 
Furthermore, $\beta_B(y_k)$ is well defined for $k \in [\ell]$, since 
$x_k \prec_{v_k} y_k$ implies $y_k \in \cl(F) \setminus B$ by \Cref{cor:<} (1) and (2). 
We define the {\em profit} of $P$ as 
\begin{align*}
    \profit_B(P) := \alpha_B(x_1) + \left( \sum_{k=1}^{\ell} w(v_k) \right) - \beta_B(y_\ell). 
\end{align*}
The following observation shows that
the profit of a promotion path $P$ can be viewed as the length of $P$, adjusted by the ``potentials'' $\alpha_B$ and $\beta_B$.
\begin{observation}\label{obs:potential}
    For a promotion path $P=(x_1, y_1, x_2, y_2, \dots , x_{\ell}, y_{\ell})$ with respect to a well-formed common base $B$,
    its profit $\profit_B(P)$ coincides with the length of $P$ in which the length of the arc $(x_k, y_k)$ ($k \in [\ell]$) is set to $\alpha_B(x_k) + w(v_k) - \beta_B(y_k)$
and that of the arc $(y_k, x_{k+1})$ ($k \in [\ell -1]$) is set to 
$\beta_B(y_{k}) - \alpha_B(x_{k+1})$.
\end{observation}

An arc $(y,x)$ of $D(B)$ with $y \in \cl(F) \setminus B$ and $x \in B$ is said to be \emph{essential}
if $\alpha_B(x) < \beta_B(y)$.
From the perspective of \Cref{obs:potential},
the arc $(y,x)$ is essential if and only if
its length, adjusted by $\alpha_B$ and $\beta_B$, is positive.
We are now ready to define an \emph{essential} promotion path.
\begin{definition}
A promotion path $P=(x_1, y_1, x_2, y_2, \dots , x_{\ell}, y_{\ell})$ with respect to $B$ is said to be \emph{essential}
if $(y_k, x_{k+1})$ is essential for every $k \in [\ell-1]$.
\end{definition}

We now claim that a well-formed common base $B$ is popular if and only if
$D(B)$ has no essential promotion path with positive profit.
For ease of later use, we state this claim in terms of the initial element of a path.

We say that $x$ is \emph{positive with respect to $B$}
if there is an essential promotion $x$-path with respect to $B$ having positive profit.
The (non)existence of a positive element can characterize the popularity in the following sense.
The proof is given in~\Cref{sec:proof:popular-improving}. 

\begin{proposition}\label{prop:popular-improving}
A well-formed common base $B$ is popular if and only if
there are no positive elements with respect to $B$.
\end{proposition}

\subsection{Characterization and Algorithm}

The following proposition shows that whether a vertex $x$ is positive does not depend on the choice of $B$, provided that $B$ contains $x$.
The proof is given in~\Cref{sec:proof:positivity}.

\begin{proposition}\label{prop:positivity}
    Let $x \in F \cup S$.
    For any $x$-including well-formed common bases $B$ and $B'$,
    the element $x$ is positive with respect to $B$
    if and only if it is positive with respect to $B'$.
\end{proposition}

We say that $x \in F \cup S$ is \emph{positive}
if it is positive with respect to some $x$-including well-formed common base. 
By~\Cref{prop:positivity}, this is equivalent to being positive with respect to every $x$-including well-formed common base, provided that at least one such well-formed common base exists.
%
Let $X_+$ denote the set of all positive elements.
Then,~\Cref{prop:popular-improving,prop:positivity} immediately imply the following corollary. 
\begin{corollary}\label{cor:populariffX+}
    A well-formed common base is popular if and only if it is $X_+$-excluding.
\end{corollary}

This, together with \Cref{ob:assumebase} and \Cref{prop:populartowellformed}, yields the following characterization.

\begin{corollary}\label{cor:sol_ind_characterization}
    A common independent set is popular if and only if it is an $X_+$-excluding well-formed common base.
\end{corollary}

It is worth noting that this is a solution-independent characterization of a popular common independent set, since $X_+$ does not depend on the choice of a common base.
By this characterization, we obtain the following algorithm for finding a popular common independent set if one exists. 




\begin{algorithm}
\caption{\textsc{PopularCommonBase}}
\label{alg:popular-common-base}
\begin{algorithmic}[1]
\Require An instance with last resorts; see \Cref{sec:lastresort}.
\Ensure A popular common independent set, if one exists.
\State Compute the sets $F$ and $S$.
\State Compute the set $X_+$ of all positive elements.
\If{there exists an $X_+$-excluding well-formed common base}
    \State return it.
\Else
    \State conclude that no popular common independent set exists.
\EndIf
\end{algorithmic}
\end{algorithm}

To show that this algorithm runs in polynomial time, we use the following propositions, whose proofs are given in Sections~\ref{sec:proof:algo:well-formed} and~\ref{sec:proof:positive-algo}, respectively.  

\begin{proposition}\label{prop:algo:well-formed}
    For any disjoint $X, Y \subseteq F \cup S$,
    we can find an $X$-including and $Y$-excluding well-formed common base in polynomial time if one exists.
\end{proposition}

\begin{proposition}\label{prop:positive-algo}
    We can compute the set $X_+$ of all positive elements in polynomial time.
\end{proposition}

By the above argument, we can prove \Cref{thm:main}. 

\begin{proof}[Proof of~\Cref{thm:main}]
As shown in \Cref{sec:lastresort}, we may assume that each $E_v$ contains the last resort.  
We prove that \textsc{PopularCommonBase} correctly finds a popular common base in polynomial time, if one exists.  
The correctness of this algorithm follows from 
\Cref{cor:sol_ind_characterization}. 

In the first line of the algorithm, we can easily compute $F$ and $S$ in polynomial time in a straightforward way. 
In the second line, we can compute $X_+$ in polynomial time by \Cref{prop:positive-algo}. 
In the third line, we can compute an $X_+$-excluding well-formed common base in polynomial time by \Cref{prop:algo:well-formed} if one exists. 
Therefore, the algorithm runs in polynomial time. 
\end{proof}

\section{Proofs of the Key Propositions}
\label{sec:proofofkeytheorems}

\subsection{Proof of~\Cref{prop:popular-improving}}
\label{sec:proof:popular-improving}

For a promotion $x$-path $P$,
let $\profit_B'(P) := \profit_B(P) - \alpha_B(x)$.
The following lemma says that 
$-\profit_B$ and $-\profit_B'$ represent
the weighted differences in preferences with $B$ and some common independent sets,
and that the existence of a promotion path with positive profit implies that of an essential one.
\begin{lemma}\label{cl:exchange}
\begin{enumerate}[label={\textup{(\arabic*)}}]
    \item For any promotion path $P$ with respect to $B$,
    there are common independent sets $I$ and $I'$ such that
$\Delta_w(B, I) = -\profit_B(P)$ and $\Delta_w(B, I') = -\profit_B'(P)$.
    \item If there is a promotion path with respect to $B$ having positive profit,
    then there is an essential promotion path with respect to $B$ having positive profit.
\end{enumerate}
\end{lemma}
\begin{proof}
(1)
Let $P = (x_1, y_1, x_2, y_2, \dots, x_\ell, y_\ell)$
be an arbitrary promotion path.
We first show that there is a common independent set $I'$ satisfying $\Delta_w(B, I') = -\profit_B'(P)$.
Let $v_{\ell+1} \in V$ satisfying $x_{\ell+1} := f(v_{\ell+1}) \in B$, $B + y_{\ell} - x_{\ell+1} \in \B$, and $\beta_B(y_{\ell}) = w(v_{\ell+1})$;
such $v_{\ell+1}$ always exists by the definition of $\beta_B$.
We now claim that no arcs of the form $(y_k, x_{k'})$ with $1 \le k' \le k \le \ell$ exist in $D(B)$.
Indeed, suppose, to the contrary, that $(y_k, x_{k'})$ with $1 \leq k' \le k \le \ell$ exists in $D(B)$.
Since $x_{k'}$ belongs to $F$ by $x_k \prec_{v_k} y_k$ and \Cref{cor:<} (1) and (2),
we have $x_k \in F$ by $k' \leq k$ and \Cref{cor:improving-dipath}.
Thus,
we obtain $w(v_{k'}) > w(v_k)$ by \Cref{cor:<}~(2),
which contradicts $w(v_{k'}) \leq w(v_k)$ by \Cref{cor:improving-dipath}.
Note that, since $D(B)$ contains the arc $(y_\ell, x_{\ell+1})$, the element $x_{\ell+1}$ is distinct from $x_1,\ldots, x_\ell$ by 
this claim. 
Therefore, in the bipartite graph $G = (X, Y; E')$ defined by 
$X = \{x_2, \dots , x_{\ell+1}\}$, $Y = \{y_1, \dots , y_{\ell}\}$, and $E' = \{\{y, x\} \mid x \in X,\ y \in Y,\ B-x+y \in \B\}$, 
there exists a unique perfect matching $\{\{y_1, x_2\}, \{y_2, x_3\}, \dots , \{y_{\ell}, x_{\ell+1}\}\}$,
which implies $(B \setminus \{ x_2, \dots, x_{\ell}, x_{\ell+1} \}) \cup \{y_1, \dots, y_\ell \} \in \B$ by \Cref{lem:uniquematching}.

Define $I' := (B \setminus \{ x_1, \dots, x_\ell, x_{\ell+1} \}) \cup \{y_1, \dots, y_\ell \}$.
Then, $I' \in \I$ by the above argument, and it is obvious that $I' \in \I_{\rm part}$,
i.e.,
$I'$ is a common independent set.
Furthermore,
since $x_k \prec_{v_k} y_k$ for each $k \in [\ell]$,
we have
\begin{align*}
    \Delta_w(B, I') = -\sum_{k = 1}^\ell w(v_k) + w(v_{\ell+1}) =- \profit_B'(P),
\end{align*}
where the second equality follows from $\beta_B(y_{\ell}) = w(v_{\ell+1})$.

We then show that there is a common independent set $I$ satisfying $\Delta_w(B, I) = -\profit_B(P)$.
If $A_B(x_1) = \emptyset$, or equivalently, $\alpha_B(x_1) = 0$,
then we obtain $\profit_B(P) = \profit_B'(P)$.
Thus, $I := I'$ is a desired common independent set.
Suppose that $A_B(x_1) \neq \emptyset$.
We pick $v_0 \in V$ satisfying $y_0 := f(v_{0}) \in A_B(x_1)$ and $\alpha_B(x_1) = w(v_{0})$
(such $v_{0}$ exists by the definition of $\alpha_B$).
Then we have $x_0 := s(v_0) \in B$,
since $y_0 = f(v_0) \in A_B(x_1) \subseteq E \setminus B$
and $B$ is well-formed.
Note that $v_{0}$ is distinct from $v_1,\ldots, v_\ell$, since $f(v_k) = x_k \in B$ for $k \in [\ell]$ by \Cref{cor:<} (2), while $f(v_0) = y_0 \not\in B$.
In this case,
we define
$I := (B \setminus \{ x_0, x_1, \dots, x_\ell, x_{\ell+1} \}) \cup \{y_0, y_1, \dots, y_\ell \}$.
By the same argument as that for $I'$,
we can see that $I$ is a common independent set satisfying
\begin{align*}
    \Delta_w(B, I) = -w(v_0) - \sum_{k = 1}^\ell w(v_k) + w(v_{\ell+1}) = -\profit_B(P),
\end{align*}
in which we use $x_0 = s(v_0) \prec_{v_0} f(v_0) = y_0$
and
$\alpha_B(x_1) = w(v_0)$.

(2)
Suppose that there is a promotion path $P = (x_1, y_1, \dots, x_\ell, y_\ell)$ with respect to $B$ having positive profit.
Among such paths, choose $P$ so that $\ell$ is minimized.
We show that each arc $(y_k, x_{k+1})$ is essential for $k \in [\ell-1]$,
i.e., $P$ is essential.

Assume to the contrary that there exists an index $k \in [\ell-1]$ such that $\beta_B(y_{k}) \le \alpha_B(x_{k+1})$.
Let $P' := (x_1, y_1, \dots, x_k, y_k)$ and $P'':= (x_{k+1}, y_{k+1}, \dots, x_\ell, y_\ell)$,
which are promotion paths with respect to $B$.
Then, by \Cref{obs:potential},
we obtain
\begin{align*}
    0 < \profit_B(P) &= \profit_B(P') + (\beta_B(y_{k}) - \alpha_B(x_{k+1})) + \profit_B(P'')\\
    &\leq \profit_B(P') + \profit_B(P''),
\end{align*}
which implies that $\profit_B(P') > 0$ or $\profit_B(P'') > 0$.
This contradicts the minimality of $\ell$.
%
\end{proof}

We are now ready to prove \Cref{prop:popular-improving}. 
The only-if part immediately follows from \Cref{cl:exchange}~(1).
Indeed, if there is an (essential) promotion path $P$ with positive profit,
then there is a common independent set $I$ satisfying 
$\Delta_w(B, I) = -\profit_B(P) < 0$ by \Cref{cl:exchange}~(1),
which implies that $B$ is not popular.

In the following, we show the if part.
Suppose that $B$ is not popular.
By \Cref{cl:exchange}~(2),
it suffices to show the existence of a (not necessarily essential) promotion path with respect to $B$ having positive profit.

By \Cref{lem:negativecycle}, 
the auxiliary digraph $D(B)$ contains a dicycle $C$ of positive weight with respect to $w_B$.
We now claim that $C$ contains at least one element with negative weight. 
Indeed, if otherwise, i.e.,
all elements in
$C = (x_1, y_1, x_2, y_2, \dots, x_\ell)$ with $x_1 = x_\ell \in B$ have nonnegative weight,
then $w_B(y_k) > 0$ holds for each $k$.
Note that $w_B(y_k) > 0$ implies $x_k \prec_{v_k} y_k$ for some $v_k \in V$.
Thus, we have $x_k \in F$ for all $k$ by \Cref{cor:<}~(1) and~(2),
which implies
$w(v_1) < w(v_2) < w(v_3) < \dots < w(v_\ell)$ by \Cref{cor:<}~(2).
This contradicts $v_\ell =v_1$.


Thus,
$C$ contains at least one element with negative weight. 
Using the elements having negative weight, the dicycle $C$ can be decomposed into a collection $\mathcal{P}$ of vertex-disjoint dipaths such that, 
for each dipath $P' \in \mathcal{P}$, the last element has negative weight, while all other elements have nonnegative weight.
This implies that the initial and last elements of each dipath are in $B$ and in $E \setminus B$, respectively.
For $P' \in \mathcal{P}$, let $V(P')$ denote the set of vertices (elements) in $P'$. 
Since 
\[
 \sum_{P' \in \mathcal{P}} \sum_{e \in V(P')} w_B(e) = \sum_{e \in V(C)} w_B(e) > 0, 
\]
there exists $P' \in \mathcal{P}$ such that $\sum_{e \in V(P')} w_B(e) > 0$. 
Let $P' = (x_1, y_1, \dots, x_\ell, y_\ell, x_{\ell+1}, y_{\ell+1})$ be such a dipath,
where $w_B(y_k) > 0$ for $k \in [\ell]$ and $w_B(y_{\ell+1}) < 0$. 
Let $v_k \in V$ denote the applicant with $x_k, y_k \in E_{v_k}$ for $k \in [\ell+1]$.
Since $w_B(y_k) > 0$ implies $x_k \prec_{v_k} y_k$ for $k \in [\ell]$,
the subpath $P := (x_1, y_1, \dots, x_\ell, y_\ell) = P' - x_{\ell+1} - y_{\ell+1}$ is a promotion path.
Moreover,
since \Cref{cor:<} (1) and (2) imply $x_{\ell +1} \in F$, 
we have $\beta_B(y_\ell) \leq w(v_{\ell+1})$ by the definition of $\beta_B(y_\ell)$.
Thus, we obtain
\begin{align*}
    \profit_B(P) \geq \profit_B'(P) = \sum_{k = 1}^\ell w(v_k) - \beta_B(y_\ell) \geq \sum_{k = 1}^{\ell} w(v_k) - w(v_{\ell+1}) = \sum_{e \in V(P')} w_B(e) > 0,
\end{align*}
which implies that $P$ is a promotion path with positive profit.
This completes the proof of \Cref{prop:popular-improving}.

\subsection{Proof of \Cref{prop:positivity}}
\label{sec:proof:positivity}
For any $y \in \cl(F)$ and $y$-excluding well-formed common base $B$,
we define
\begin{align*}
    C^*(B, y) := \{ x \in C(B, y) - y \mid \alpha_B(x) < \beta_B(y) \},
\end{align*}
which is the set of the heads $x$ of essential arcs of the form $(y,x)$ with respect to $B$.
The following lemma says that $\alpha_B(x)$, $\beta_B(y)$, and $C^*(B, y)$ are independent of 
the choice of an $x$-including/$y$-excluding well-formed common base $B$.
\begin{lemma}\label{lem:xy}
\begin{enumerate}[label={\textup{(\arabic*)}}]
    \item Let $x \in F \cup S$.
    For any $x$-including well-formed bases $B$ and $B'$,
    we have $\alpha_B(x) = \alpha_{B'}(x)$.
    \item Let $y \in \cl(F)$.
    For any $y$-excluding well-formed bases $B$ and $B'$,
    we have $\beta_B(y) = \beta_{B'}(y)$ and $C^*(B, y) = C^*(B', y)$.
\end{enumerate}
\end{lemma}

\begin{proof}
(1)
It suffices to show that, for any $x$-including well-formed base $B$,
the value $\alpha_B(x)$ coincides with
\begin{align*}
    \alpha(x) :=
    \begin{cases}
         \max\{w_i \mid i \in [p],\ x \in F_{\le i},\ x \in \cl (F_{\le i} - x) \} & \text{if such $i$ exists},\\
         0 & \text{if such $i$ does not exist},
    \end{cases}
\end{align*}
which is determined independently of $x$-including well-formed bases.
Note that $\alpha(x) > 0$ only if $x \in F$.

We first prove that $\alpha_B(x) \ge \alpha(x)$. 
If $\alpha(x) = 0$, then the inequality is obvious. 
Suppose that $\alpha(x) > 0$ holds and let $i \in [p]$ be the index such that $\alpha(x) = w_i$, 
$x \in F_{\le i}$, and $x \in \cl (F_{\le i} - x)$. 
Note that $x \in F_{\le i}$ and $x \in \cl (F_{\le i} - x)$ imply ${\rm rank}(F_{\le i} - x) =  {\rm rank}(F_{\le i})$. 
By Remark~\ref{rem:clequivalent}, we obtain 
\[
{\rm rank}(F_{\le i} - x) =  {\rm rank}(F_{\le i}) = |B \cap F_{\le i}| > |(B \cap F_{\le i}) - x|, 
\]
which implies that $F_{\le i}-x$ contains an independent set of $\M$ that is larger than $(B \cap F_{\le i}) - x$. 
Therefore, there exists $y \in F_{\le i} \setminus B$ such that $(B \cap F_{\le i}) - x + y \in \I$. 
By the augmentation property of matroids, we can add elements of $B$ to $(B \cap F_{\le i}) - x + y$ 
to obtain an independent set of size $|B|$. 
However, since 
\[
|(B \cap F_{\le i}) - x + y| = |B \cap F_{\le i}| = {\rm rank}(F_{\le i}), 
\]
we cannot add elements from $F_{\le i}$. 
Therefore, we need to add all the elements in $B \setminus F_{\le i}$, that is, 
\[
((B \cap F_{\le i}) - x + y) \cup (B \setminus F_{\le i}) = B - x + y
\]
is an independent set (and hence also a base) of $\M$. 
Since $y \in F_{\le i} \setminus B$, there exists $v \in V_{\le i}$ such that $f(v) = y$ and $w(v) \ge w_i = \alpha(x)$. 
Therefore, $\alpha_B(x) \ge w(v) \ge \alpha(x)$ holds. 

We next prove that $\alpha(x) \ge \alpha_B(x)$. 
If $\alpha_B(x) = 0$, then the inequality is obvious. 
Suppose that $\alpha_B(x) > 0$ holds and let $v \in V$ be an applicant such that $\alpha_B(x) = w(v)$, $f(v) \not\in B$, and
$B - x + f(v) \in \I$. Let $i \in [p]$ be the index with $v \in V_i$. 
Since it holds that  
\begin{align*}
    {\rm rank}(F_{\le i}) &\ge {\rm rank}(F_{\le i}-x) \ge |(B - x + f(v)) \cap F_{\le i}| \\ 
                          &= |B \cap F_{\le i}| - |\{x\} \cap F_{\le i}| + 1 \ge |B \cap F_{\le i}| = {\rm rank}(F_{\le i})
\end{align*}
by Remark~\ref{rem:clequivalent}, 
all the inequalities are tight.  
In particular, ${\rm rank}(F_{\le i}) = {\rm rank}(F_{\le i}-x)$ and $|\{x\} \cap F_{\le i}| = 1$ hold, and hence $x \in F_{\le i}$ and $x \in \cl (F_{\le i} - x)$. 
Therefore, $\alpha(x) \ge w_i = w(v) = \alpha_B(x)$ holds. 

(2)
Take any $y$-excluding well-formed common bases $B$ and $B'$.
By the definition~\eqref{eq:beta_B} and \Cref{lem:lastarc},
we have $\beta_B(y) = \beta_{B'}(y) = w_{i^*}$,
where $i^*$ denotes the minimum index such that $y \in \cl(F_{\leq i^*})$.
This implies the former assertion of~(2).

We then show that $C^*(B, y) \subseteq C^*(B', y)$;
the reverse inclusion $C^*(B, y) \supseteq C^*(B', y)$ follows from symmetry,
which implies the latter assertion of~(2).

Take any $x \in C^*(B, y) \subseteq C(B, y) - y$,
i.e., $B - x + y \in \B$ and $\alpha_B(x) < \beta_B(y)$;
our aim is to prove that $x \in C^*(B', y)$.
By \Cref{lem:lastarc},
we have $x \in F$.
By the simultaneous exchange property of matroid bases for $B'$ and $B - x + y$ with respect to $y$, there exists $x' \in B' \setminus (B - x + y)$ such that
\begin{align}
    &B' + y - x' \in \B,\label{eq:B'}\\
    &(B - x + y) - y + x' = B - x + x' \in \B\label{eq:B''}.
\end{align}
Note that $x'$ belongs to $C(B', y) - y$ by~\eqref{eq:B'}, implying that $x' \in F$ by \Cref{lem:lastarc}.

We now claim that $x' = x$; otherwise, 
letting $v' \in V$ denote the applicant such that $x' = f(v')$,
we obtain
\begin{align*}
    w(v') \leq \alpha_B(x) < \beta_B(y) = \beta_{B'}(y) \leq w(v'),
\end{align*}
a contradiction.
Here, the first inequality follows from~\eqref{eq:B''}, 
$x' \in F$, and $x' \neq x$ (implying $x' \in A_B(x)$)
and the definition~\eqref{eq:alpha_B} of $\alpha_B$,
the second (strict) inequality follows from $x \in C^*(B, y)$,
the equality follows from the former assertion of~(2),
and the last inequality follows from~\eqref{eq:B'} with $x' \in F$ and the definition~\eqref{eq:beta_B} of $\beta_{B'}$.
Since $x' = x$, we obtain $B'+y-x \in \B$ by (\ref{eq:B'}), that is, $x \in C(B', y)-y$. 
Therefore,
$x \in C^*(B', y)$ follows from $\alpha_{B'}(x) = \alpha_B(x) < \beta_B(y) = \beta_{B'}(y)$,
where the first and last equalities follow from the assertion~(1) and the former assertion of~(2),
respectively.
\end{proof}

We are now ready to show \Cref{prop:positivity}.
Let $x \in F \cup S$ and take any $x$-including well-formed common bases $B$ and $B'$.
It suffices to prove that
\begin{itemize}
    \item the set of all essential promotion $x$-paths with respect to $B$ is the same as that with respect to $B'$; and
    \item for any essential promotion $x$-path $P$ with respect to $B$ (or equivalently, with respect to $B'$), we have $\profit_B(P) = \profit_{B'}(P)$.
\end{itemize}

Take an arbitrary essential promotion $x$-path
$P = (x = x_1, y_1, x_2, y_2, \dots, x_\ell, y_\ell)$ with respect to $B$,
where $v_k \in V$ denotes the applicant with $x_k \in E_{v_k}$
for $k \in [\ell]$.
Since $x_1 = x \in B'$ and $x_1 \prec_{v_1} y_1$,
we have 
$y_1 \notin B'$.
Thus, the arc $(x_1, y_1)$ exists also in $D(B')$.
In addition, since $(y_1, x_2)$ is an essential arc in $D(B)$,
we obtain $x_2 \in C^*(B, y_1) = C^*(B', y_1)$,
where the equality follows from \Cref{lem:xy}~(2).
This implies that the arc $(y_1, x_2)$ exists also in $D(B')$
and is essential. 
In particular, $x_2 \in B'$. 
By repeating this argument,
we can conclude that $P$ is an essential promotion $x$-path with respect to $B'$. 
Therefore, the set of all essential promotion $x$-paths with respect to $B$ is included in that with respect to $B'$.
By symmetry, the reverse inclusion also holds.
Thus, we obtain the first bullet.

For any essential promotion $x$-path $P$ with respect to $B$ (or equivalently, with respect to $B'$),
we have
\begin{align*}
\profit_B(P) &= \alpha_B(x_1) + \sum_{k = 1}^\ell w(v_k) - \beta_B(y_\ell)\\
&= \alpha_{B'}(x_1) + \sum_{k = 1}^\ell w(v_k) - \beta_{B'}(y_\ell)\\
&= \profit_{B'}(P),
\end{align*}
in which the second equality follows from \Cref{lem:xy}.
Thus, we obtain the second bullet.

This completes the proof of \Cref{prop:positivity}. 

\subsection{Proof of \Cref{prop:algo:well-formed}}
\label{sec:proof:algo:well-formed}
Our goal is to compute a common base $B$ such that
\begin{align}
    &X \subseteq B \subseteq (F \cup S) \setminus Y, \label{eq:inc-exc}\\
    &|B \cap F_{\le i}| = \rank(F_{\le i}) \qquad (i \in [p]). \label{eq:rank}
\end{align}
Indeed, by \Cref{def:wellformed} and \Cref{rem:clequivalent},
a common base $B$ is well-formed if and only if $B \subseteq F \cup S$ and $B$ satisfies the condition~\eqref{eq:rank}.


We define $\theta, \varphi \in \Z^E$ by
\begin{align*}
    \theta(e) :=
    \begin{cases}
        -1 & \text{if $e \in Y \cup (E \setminus (F \cup S))$}, \\
        p+1-i & \text{if $e \in F_{i} \setminus Y$ with $i \in [p]$}, \\
        0 & \text{otherwise},
    \end{cases}
    \qquad
    \varphi(e) :=
    \begin{cases}
    1 & \text{if $e \in X$},\\
    0 & \text{otherwise},
    \end{cases}
    \qquad (e \in E).
\end{align*}
We claim that a common base $B$ satisfies the conditions~\eqref{eq:inc-exc} and~\eqref{eq:rank}
if and only if
\begin{align*}
    \sum_{e \in B} (\theta(e) + \varphi(e)) = \sum_{i=1}^p {\rm rank}(F_{\le i}) + |X|.
\end{align*}
Indeed, by a direct computation, for any common base $B$, we have that
\begin{align*}
    \sum_{e \in B} (\theta(e) + \varphi(e)) &= \left(\sum_{i=1}^p (p+1-i) \cdot |B \cap (F_{i} \setminus Y)| - |B \cap Y| - |B \setminus (F \cup S)|\right) + |B \cap X|\\
&\leq  \sum_{i=1}^p (p+1-i) \cdot |B \cap F_{i}| + |X|\\
&=  \sum_{i=1}^p |B \cap F_{\le i}| + |X|\\
&\leq  \sum_{i=1}^p \rank(F_{\le i}) + |X|. 
\end{align*}
Here, the first inequality is tight if and only if \eqref{eq:inc-exc} holds,
and
the last inequality is tight if and only if \eqref{eq:rank} holds. 

Consequently, if there exists a well-formed common base satisfying \eqref{eq:inc-exc} and \eqref{eq:rank}, then any common base maximizing
$\sum_{e \in B} (\theta(e) + \varphi(e))$ must satisfy these conditions.  
Thus, it suffices to compute a maximum-weight common base with respect to $\theta + \varphi$. 
This can be done in polynomial time using a weighted matroid intersection algorithm; see e.g.,~\cite[Theorem 41.8]{Schrijver2003}, 
which completes the proof of \Cref{prop:algo:well-formed}.

\subsection{Proof of \Cref{prop:positive-algo}}
\label{sec:proof:positive-algo}


Fix an arbitrary $x^* \in F \cup S$.
It suffices to show that we can determine if $x^*$ is positive in polynomial time.

First, we check if there is an $x^*$-including well-formed common base and compute it if one exists;
this can be done in polynomial time by \Cref{prop:algo:well-formed}.
If there is no $x^*$-including one, then we conclude that $x^*$ is not positive.

Suppose that we obtain an $x^*$-including well-formed common base $B$.
Then, we construct the bipartite digraph $D^*(B)$ with arc lengths as follows.
The vertex set is $E$, with bipartition $(B, E \setminus B)$, 
and the arc set consists of
\begin{itemize}
    \item the arcs $(x, y)$ of length $\alpha_B(x) + w(v) - \beta_B(y)$ for $x \in B$ and $y \in E \setminus B$ such that $x \prec_{v} y$ for some $v \in V$; and
    \item the arcs $(y, x)$ of length $\beta_B(y) - \alpha_B(x)$ for $x \in B$ and $y \in E \setminus B$ such that $x \in C(B, y) - y$ and $\alpha_B(x) < \beta_B(y)$.
\end{itemize}
Note that the arcs $(y, x)$ in the second bullet are exactly the essential arcs.
Hence, $D^*(B)$ (without arc length) is the subgraph of $D(B)$ consisting of the arcs that can appear in some essential promotion path with respect to $B$.

It is clear from the construction of $D^*(B)$ that any dipath $P = (x = x_1, y_1, x_2, y_2, \dots, x_\ell, y_\ell = y)$ in $D^*(B)$ from $x \in B$ to $y \in E \setminus B$ is an essential promotion path with respect to $B$,
and that the length of $P$ coincides with $\profit_B(P)$; see \Cref{obs:potential}.
Thus, by \Cref{prop:positivity}, $x^*$ is positive if and only if
there is a dipath of positive length in $D^*(B)$ that starts at $x^*$ and ends with some $y \in E \setminus B$.

Furthermore, we can see that $D^*(B)$ forms a directed acyclic graph (DAG).
Indeed, suppose to the contrary that $D^*(B)$ contains a directed cycle $(x_1, y_1, x_2, y_2, \dots, x_\ell, y_\ell, x_{\ell +1} = x_1)$,
where $x_k \prec_{v_k} y_k$ for $k \in [\ell]$.
Since $x_k \prec_{v_k} y_k$ and $x_{k+1} \in C(B, y_k) - y_k$ for all $k \in [\ell]$,
we have $x_2, \dots, x_{\ell}, x_{\ell + 1} = x_1 \in F$
by \Cref{cor:<} (1) and (2).
Since $x_1, \dots , x_{\ell} \in F$, we obtain 
$w(v_1) < w(v_2) < \cdots < w(v_\ell) < w(v_{\ell+1}) = w(v_1)$ by \Cref{cor:<}~(2),
a contradiction.

Clearly, we can construct $D^*(B)$ in polynomial time.
Since $D^*(B)$ is a DAG, we can solve the longest path problem for $D^*(B)$ in polynomial time.
Thus, we can check if $x^*$ is positive in polynomial time,
which completes the proof of \Cref{prop:positive-algo}.

\section{Concluding Remarks}
\label{sec:conclusion}

In this paper, we presented the first polynomial-time algorithm 
for the weighted popular common independent set problem. 
As special cases, our result yields the first polynomial-time algorithms for the weighted popular matching problem 
with matroid constraints and the weighted popular branching problem. 
The main technical contribution of this work is 
the introduction of the notion of an essential promotion path 
and the proof that its existence is independent of the choice of the common base $B$. 
This solution-independent characterization serves as the foundation of our algorithmic framework.

We remark that, given costs $c(e) \in \mathbb{R}$ for each $e \in E$, 
a popular common base $B$ minimizing $\sum_{e \in B} c(e)$ can also be computed in polynomial time using the above characterization. 
Indeed, by a slight modification of \Cref{prop:algo:well-formed}, 
we can compute, in polynomial time, an $X$-including and $Y$-excluding well-formed common base $B$ 
that minimizes $\sum_{e \in B} c(e)$; 
this can be achieved by computing a common base maximizing
$\sum_{e \in B} (K\theta(e) + K\varphi(e) - c(e))$,
where $\theta$ and $\varphi$ are as defined in the proof of \Cref{prop:algo:well-formed}, 
and $K$ is a sufficiently large number. 
Therefore, we can compute an $X^+$-excluding well-formed common base $B$ that minimizes $\sum_{e \in B} c(e)$, 
which is the desired popular common base.

This paper considers only strict preference orders. 
A natural direction for future research is to extend the model to more general preference structures, 
such as preferences with ties or arbitrary partial orders. 
Such extensions have been studied extensively in previous work. 
For example, Mestre \cite{Mes06,Mes14} gave a polynomial-time algorithm 
for the weighted popular matching problem in the presence of ties, 
and Kavitha, Makino, Schlotter, and Yokoi~\cite{KMSY25} gave a polynomial-time algorithm 
for the (unweighted) popular common independent set problem under arbitrary partial-order preferences. 
The question of whether these results can be extended to 
the weighted popular common independent set problem remains an interesting open problem.
Investigating weighted generalizations of 
popular assignments \cite{KKMSS22a} and popular common bases \cite{KMSY25}, 
and a matroidal generalization of 
weighted popular matchings with two-sided preferences \cite{HC24} is also of interest.

\section*{Acknowledgments}
This work was partially supported by the joint project of Kyoto University
and Toyota Motor Corporation, titled “Advanced Mathematical Science for Mobility Society”, 
by JST ERATO Grant Number JPMJER2301, 
and by JSPS KAKENHI Grant Numbers
JP22K17854, 
JP24K02901, 
JP24K14828, 
JP24K21315, 
JP26K21777, 
JP26K21171, 
Japan.

\bibliographystyle{plain}
\bibliography{ref}

\appendix

\section{Basic Lemma for Matroids}


\begin{lemma}\label{lem:remark42}
 Let $B \in \I$ and let $X \subseteq E$. Then the following are equivalent: 
 \begin{enumerate}[label={\textup{(\roman*)}}]
     \item $X \subseteq \cl(B\cap X)$, 
     \item $\cl(X) = \cl(B\cap X)$, 
     \item $|B \cap X| = {\rm rank}(X)$.      
 \end{enumerate} 
\end{lemma}

\begin{proof}
If $X \subseteq \cl(B\cap X)$, then 
\[
\cl(X) \subseteq \cl(\cl(B\cap X)) = \cl(B\cap X) \subseteq \cl(X), 
\]
and hence $\cl(X) = \cl(B\cap X)$. Therefore, (i) implies (ii). 
Next, suppose that $\cl(X) = \cl(B\cap X)$.
Since taking closure does not change rank, we have
\[
{\rm rank}(X)={\rm rank}(\cl(X))={\rm rank}(\cl(B\cap X))={\rm rank}(B\cap X).
\]
Since $B\cap X$ is independent, ${\rm rank}(B\cap X)=|B\cap X|$. Hence
(ii) implies (iii). 
Finally, suppose that $|B\cap X|={\rm rank}(X)$.
Since $B \cap X$ is an independent set contained in $X$, condition (iii) is equivalent to that
$B \cap X$ is an inclusionwise maximal independent set contained in $X$. 
This is equivalent to the condition (i). 

Therefore, all the three conditions are equivalent. 
\end{proof}

\end{document}